\documentclass{IEEEtran}
     \usepackage{graphicx,amssymb,amstext,amsmath,amsthm}
    \usepackage{epsfig}
    \usepackage{multirow}
    \usepackage{graphics}
    \usepackage{graphicx}
    \usepackage{psfrag}
    \usepackage{cite}
    \usepackage{color}
    \usepackage{float}
    \usepackage[noend]{algorithmic}
    \usepackage{algorithm}
    \usepackage{subfig}
    \usepackage{grffile}
    \usepackage{alltt}
    \usepackage{array}
    \usepackage{times}
    \usepackage{url}
    \usepackage{mathptmx} 
    \usepackage{paralist}       
    \usepackage[bottom]{footmisc}

    \usepackage{url}
    \theoremstyle{definition}

    \newtheorem{lemma}{Lemma}

\begin{document}\fontsize{10}{11.7}\rm
    \IEEEoverridecommandlockouts 

    \title{On the Transient Behavior of CHOKe}
    \author{%
    Addisu Eshete and Yuming Jiang\\
    \authorblockA{%
    Centre for Quantifiable Quality of Service in Communication Systems\authorrefmark{1}\thanks{\textsuperscript{\footnotesize\ensuremath{*}}``Centre for Quantifiable Quality of Service in Communication Systems, Centre of Excellence'' appointed by The Research Council of Norway, funded by the Research Council, NTNU and UNINETT. http://www.q2s.ntnu.no} and Department of Telematics\\ Norwegian University of Science and Technology, Trondheim, Norway\\
    \texttt{addisu.eshete@q2s.ntnu.no \:\: ymjiang@ieee.org}
      }
    }
    \maketitle

    \begin{abstract}
CHOKe is a simple and stateless active queue management (AQM) scheme. Apart from low operational overhead, a highly attractive property of CHOKe is that it can protect responsive TCP flows from unresponsive UDP flows. Particularly, previous works have proven that CHOKe is able to bound both bandwidth share and buffer share of (a possible aggregate) UDP traffic (flow) on a link. However, these studies consider, and pertain only to, a steady state where the queue reaches equilibrium in the presence of many (long-lived)  TCP flows and an unresponsive UDP flow of fixed arrival rate. If the steady state conditions are perturbed, particularly when UDP traffic rate changes over time, it is unclear whether the protection property of CHOKe still holds. Indeed, it can be examined, for example, that when UDP rate suddenly becomes 0 (i.e., flow stops), the unresponsive flow may assume close to full utilization in sub-RTT scales, potentially starving out the TCP flows. To explain this apparent discrepancy, this paper investigates CHOKe queue properties in a transient regime, which is the time period of transition between two steady states of the queue, initiated when the  rate of the unresponsive flow changes. Explicit expressions that characterize flow throughputs in transient regimes are derived. These results provide additional understanding of CHOKe, and give some explanation on its intriguing behavior in the transient regime. 

    \end{abstract}


\section{Introduction}\label{sec:intro}

\subsection{Overview on Flow Protection}

Broadly speaking, there are two distinct but complementary ways to enforce flow fairness and protection in the Internet. Following the end-to-end architectural design principle of the Internet ~\cite{e2eArgument84}, the more classical way has been via congestion control algorithms~\cite{RFC5681}. These algorithms are typically implemented in the transport protocols (e.g., TCP) of end hosts. To ensure \emph{global} fairness, such schemes require all users to adopt them and respond to network congestion properly. However, this requirement can hardly be met for at least two reasons. Firstly, there is no performance incentive to end users. This is because users who lack the congestion control algorithms, intentionally or otherwise, may end up with a lion share of bandwidth. Secondly, in order to meet real-time requirements, many applications do not implement congestion control. Hence, to protect responsive (e.g., TCP) flows from unresponsive (e.g., UDP) ones, solely relying on the end-to-end schemes can be unfair or risky, and it is necessary to introduce some mechanisms in the network. This motivates the second approach to fairness and protection. 

The second approach is provided through router mechanisms. Such a router mechanism can be either (1) per flow fair queueing (PFFQ) scheme,  e.g., Weighted Fair Queueing (WFQ) \cite{demersWFQ89}, or (2) queue management (QM) scheme, e.g., Random Early Detection (RED) ~\cite{REDfloyd93,adaptiveRED2001}. PFFQ schemes share link bandwidth among flows in a fair manner. Typically, they isolate flows into separate logical or physical FIFO queues and maintain flow-level state information. By building firewalls around heavy users, flow isolation protects well-behaved flows and enables performance guarantees to such flows~\cite{keshavBook97}. Nevertheless, the maintenance of per flow state and the dynamic management of complex queue structure are widely believed to be problematic\footnote{Proponents of per flow architecture argue that the number of flows requiring scheduling at any moment is a lot less than the total number traversing the router. Hence, it is possible to realize scalable implementation of the per flow architecture~\cite{kortebiSIGMETRICS05, flowBufferJSAC99, esheteAFpFT11, perFlowDQS07}.} for high speed implementation.

Queue management (QM) schemes specialize in buffer allocation and are normally significantly simpler in design, typically with a single queue shared by all flows. Among QM schemes, RED~\cite{REDfloyd93,adaptiveRED2001} is probably the most widely known scheme. It maintains an exponentially moving average queue size that indicates the level of congestion in the router. A congested RED router drops incoming packets with a certain probability dependent on the queue size. Since the dropping probabilities are applied globally to all flows, both high rate and low rate flows can be punished in equal measures. In fact, based on the nature of Internet flows (e.g., flow sizes, underlying transport protocol) and resultant differences in their responsiveness to congestion, the same ambient drop rate of RED can be more detrimental and highly unfair to some flows. To deal with this unfairness, more complex variants of RED, e.g.,  Flow RED (FRED)~\cite{fred97}, RED with Proportional Differentiation (RED-PD)~\cite{ratulICNP01}, have been proposed to apply differential per flow drop rates. However, such schemes typically need to maintain \emph{partial flow state} to be able to discriminate drop rates among flows.

With no flow isolation, fairness afforded by a QM scheme is generally approximate. For example, a recent simulation study consisting entirely of UDP flows shows that the bandwidth allocated by FRED to congested UDP flows can differ by several factors~\cite[see Fig.~6]{esheteAFpFT11}. This begs the question: \textbf{How much of the bandwidth, or generally the shared resources, can an unresponsive / rogue flow steal in QM schemes?}

\subsection{CHOKe}\label{subsec:introChoke}

A highly novel QM scheme which, unlike FRED or RED-PD, does not require flow state to be maintained in the router, is CHOKe~\cite{chokeInfocom00}. CHOKe is proposed to protect rate-adaptive (responsive) flows from unresponsive ones. It uses the recent admissions, i.e. packets currently queued in the buffer, to penalize the high bandwidth flows. It can be implemented by a few tweaks of the RED algorithm. Specifically,\emph{``when a packet arrives at a congested router, CHOKe draws a packet at random from the FIFO buffer and compares it with the arriving packet. If they both belong to the same flow, then they are both dropped; else the randomly chosen packet is left intact and the arriving packet is admitted into the buffer with a probability (based on RED) that depends on the level of congestion~\cite{chokeInfocom00}.''}



A promising property of CHOKe is that it provides analytically proven protection of responsive flows from an unresponsive flow at the congested router, which provides an answer to the aforementioned question. Specifically, in the presence of many flows, the following interesting and peculiar steady-state properties of CHOKe have been derived~\cite{ChokeToN04,ChokeSigmetrics03, allertonCHOKe01}:
\begin{itemize}
    \item \emph{Limits}: An unresponsive UDP flow cannot exceed certain limits in buffer share and link bandwidth share \cite{ChokeSigmetrics03, allertonCHOKe01}. The maximum UDP bandwidth share is $(e+1)^{-1}=26.9\%$ of link capacity, and the maximum buffer share is $50\%$.
    \item \emph{Asymptotic property}: As the UDP rate increases without bound, its buffer share can asymptotically reach $50\%$ and queueing delay can be reduced by \emph{half}, but its link utilization drops to \emph{zero}. 
    \item \emph{Spatial distribution}:  The spatial packet distribution in the queue can be highly nonuniform~\cite{ChokeToN04}. The probability of finding  a packet belonging to a high rate flow  in the queue diminishes dramatically as we move towards the head of the queue. Correspondingly, the flow distribution in queue is skewed with most packets of high rate flows found closer to queue tail while packets of low rates are found closer to queue head. 
\end{itemize}

    \subsection{Contribution of This Paper}

    To the best of our knowledge, previous analytical studies on CHOKe \cite{ChokeToN04, ChokeSigmetrics03, allertonCHOKe01} are restricted to the steady state where the traffic rate of the UDP flow is assumed constant.  However, this assumption is too restrictive, limiting more in-depth understanding of CHOKe. This paper studies CHOKe behavior in the face of dynamically changing UDP rates and at the same time generalizes the steady-state properties proved in the earlier works. Particularly, this paper investigates CHOKe queue properties in a transient regime, which is the time period of transition between two steady states of the queue when the rate of the unresponsive flow changes. 


    From modeling perspective, the study of the transient behavior of CHOKe is an arduous task for two main reasons: (a) {\em leaky} nature of the queue, meaning that packets already in queue may be dropped later, (b) continuous state transition of the queue in the transient regime. Due to (a), the delay of a packet is not merely the backlog the packet sees upon arrival divided by the link capacity as in non-leaky queues. Besides, the spatial packet distribution of a flow in the queue can be nonuniform throughout the queue. Due to (b), many parameters that characterize the queue (e.g., flow matching probability, backlog size, skewed packet distributions of flows in the queue) are likely to be dynamically changing. Both constraints prohibit us from making ``safe'' simplifying assumptions in analyzing the transient behavior of CHOKe. 


    In this paper, we take the first step in characterizing the transient UDP {\em transmission} rates at a CHOKe queue in the immediate aftermath of change in the UDP traffic {\em arrival} rate. In particular, we focus on: (1) how UDP utilization evolves in transient time, and more importantly (2) the limits, i.e., how far the utilization goes up or down in transient time. We notice that as the UDP traffic arrival rate goes up or down, its transient transmission rate can go in opposite direction in a dramatic fashion. For example, when the UDP rate sharply decreases, its utilization rapidly soars and often exceeds the steady state limits asserted above. Extreme transient behaviors are observed when a very high rate UDP flow abruptly stops. In such cases, UDP transient utilization can suddenly jump from 0\% to over, say, 70\%. This intriguing phenomenon cannot be explained with the literature results.

The contribution of this paper is several-fold. First, the above intriguing phenomenon is illustrated with examples, which forms a motivation of the work. Second, for any given UDP traffic rate, we  derive quantitatively the queue parameters that characterize the spatial properties of the queue in a steady state. Third, by leveraging the queue parameters above and abstracting the UDP rate change by a factor, both the evolution and the extreme points of UDP throughput in transient regime are analytically derived for any arbitrary UDP rate change. Last but not least, we obtain generic plots that succinctly represent both transient and steady state UDP throughput behaviors. Extensive simulations confirm the validity of the theoretical results. 


    \subsection{Structure of paper}
    The rest of this paper is organized as follows. Next section presents the system setup, basic assumptions and notation. Sec.~\ref{sec:motivationBackgnd} explains the motivation using examples, and presents some background on steady state behavior. Since the transient phase is a transition period between two steady states, discussion on steady state models is as relevant. Our theoretical models are presented in Sec.~\ref{sec:modifiedSDM}. In particular, Sec.~\ref{subsec:rateConservation} derives the insightful \emph{rate conservation argument} and obtains  further simplifying assumptions required for the transient analysis; Sec.~\ref{subsec:modifiedSDM} lays out the theoretical foundation on the \emph{spatial distribution model} just before rate change; Sec.~\ref{subsec:transient} tracks the UDP link utilizations and derives its properties during the transient period. Sec.~\ref{sec:evaluation} presents model validation and simulation results. Sec.~\ref{sec:conclusion} concludes the paper.

\section{System Model and Notation} \label{sec:modelAssumptionNotation}
    The studied system is shown in Fig.~\ref{fig:sysModel}, where $N$ rate-adaptive similar TCP flows share a link with a single unresponsive / aggressive UDP flow. Flows are indexed from $0,~\cdots,~N$, where $0$ denotes the UDP flow. Since TCP flows are similar, hereafter $1$ denotes a typical TCP flow. Packets are queued and scheduled in the FIFO manner on to the link. The steady-state backlog size in packets is denoted by $b$.  We assume large $N$ and $b$.
        \begin{figure}[h]
            \centering
            \includegraphics[width=0.4\textwidth]{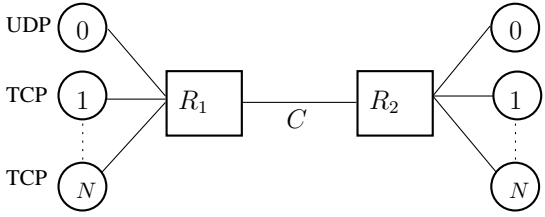}
            \caption{System model.}
            \label{fig:sysModel}
        \end{figure}

 The full notation is summarized in Table~\ref{tab:notation}. The key performance metric is the flow (link) utilization denoted by $\mu_i,~i\in[0,1]$. It represents the flow's share of bandwidth on the link. Given UDP link utilization $\mu_0$ and link capacity $C$,  the throughput of all flows  can be computed.

 A queue parameter that is not indexed with a control variable, say time $t$, designates the value of the parameter in the steady state. Otherwise, the parameter is dependent on that control variable. For example, $\mu_0$ designates UDP utilization in a steady state, while $\mu_0(t)$ is changing with $t$, more likely during a transient regime.

        \begin{table}[h!]\centering
         \caption{Notation}
         \label{tab:notation}
         \begin{small}
         \begin{tabular}{| l l |}
         \hline
           Symbol &     Description    \\ [3pt]
           \hline
            $N$      & number of TCP flows                                          \\
            $r$     & congestion(RED)-based dropping probability, or              \\
                & ambient drop probability (common to all flows)                    \\
            $x_i$   & source rate of flow $i$                                       \\
            $\mu_i$ & link utilization of flow $i$                                  \\
            $b_i$   & amount of flow $i$ packets in buffer                                    \\
            $b$     & total backlog  $b= \sum_{i=0}^{N} b_i $  in packets           \\
            $h_i$   & the ratio $b_i/b$ (matching probability)                      \\

           $y\in[0,b]$   &position in queue                                     \\
           $v(y)$       & packet velocity at $y$                                    \\
           $\tau(y)$    & queueing delay to reach at $y$                            \\
           $\rho_i(y)$     & prob. of finding flow $i\in[0,1]$ at position $y$ \\
           \hline
    \end{tabular}
    \end{small}
    \end{table}

    A consequence of large $N$ is that for a TCP flow, the packet matching in CHOKe is rare and its drop rate is mainly due to congestion (RED). That is,
    \begin{align}\label{eq:tcpNoMatch}
        h_1 &=\frac{b_1}{b}=\frac{b_1}{b_0+Nb_1}\leq \frac{1}{N} \approx 0
    \end{align}

\section{Motivation and Background}\label{sec:motivationBackgnd}
    \subsection{Motivating Examples}\label{subsec:motivation}
    To clarify our motivation, we provide example scenarios employing the network setup depicted in Fig.~\ref{fig:sysModel}. There are $N=100$ TCP sources and a UDP flow whose arrival rate is dynamically varying. The simulation parameters are described in full in Sec. \ref{sec:evaluation}.

    \emph{Example 1: }
The initial UDP arrival rate $x_0$ was $0.5C$ and $0.25C$, where $C$ is the link capacity, for Experiment 1 and Experiment 2, respectively. At $t=21$, the $x_0$ suddenly jumps by factors of 4 and 12 to $2C$ and $3C$, respectively, and then returns back to $0.5C$ and $0.25C$ at $t=22$. We conduct 500 replications of the two experiments and the resulting UDP flow utilizations (as measured using time intervals of $10$ms) are shown in Fig.~\ref {fig:4x12xUtil}.

      \begin{figure}[h]\centering
           \includegraphics[width=0.4\textwidth]{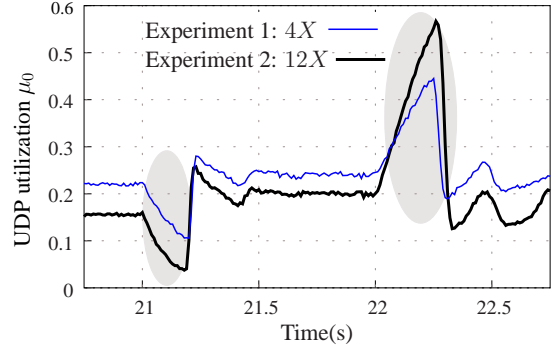}
        \caption{Transient UDP utilizations when rate $x_0$ changes by a factor of 4 and 12.}
        \label{fig:4x12xUtil}
      \end{figure}

   \emph{Example 2: }
The initial UDP arrival rate was $10C$. As shown in Fig.~\ref{fig:x0flap1C10C}, the input UDP rate flaps between $1C$ and $10C$ every 250ms. The figure shows UDP utilization averaged over 1000 replications, with measurements taken every 1ms. 

      \begin{figure}[h]\centering
        \includegraphics[width=0.4\textwidth]{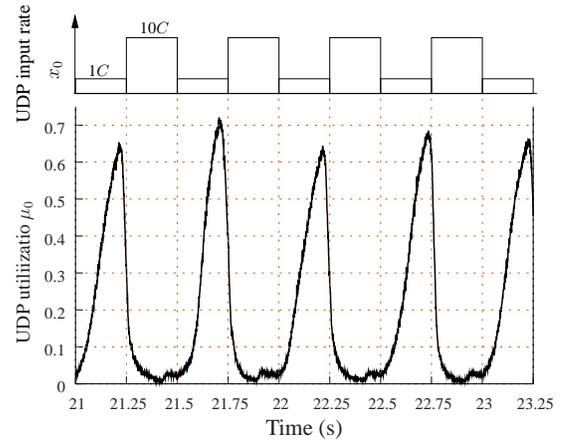}
        \caption{Transient UDP utilizations when $x_0$ flaps between $1C$ and $10C$ every 250ms.}
        \label{fig:x0flap1C10C} 
      \end{figure}


\subsection{Observation and Objective}\label{subsub:exObservation}
    In this paper, we are interested in the peculiar behaviors  (e.g., shaded in Fig.~\ref{fig:4x12xUtil}) when UDP arrival rate varies. The behaviors can represent transitions of a CHOKe queue from one steady-state to another and typically lasts a full queueing delay. We note during this transient phase that,
    \begin{itemize}
        \item The queue does not simply undergo a \emph{smooth} transition between the two steady states.  Surprisingly, the change in transient UDP utilization is not determined by the $\mu_0$ values of the two steady states. Rather, it follows the direction opposite to the arrival rate change. Specifically, if UDP arrival rate goes up, its utilization dips and vice-versa. For example, consider Experiment 2 shown in Fig.~\ref{fig:4x12xUtil}. Instead of steadily increasing  from a steady state value $\mu_{0,0.25C}=16\%$  to another steady state value $\mu_{0,3C}=21\%$, the transient UDP link utilization first whittles down to 3.75\% at $t=21.18$s. (See Fig.~\ref{fig:olm-steadystate} to find out steady state $\mu_0$ values.)
        \item While UDP buffer and utilization bounds stipulated in earlier works~\cite{ChokeToN04,ChokeSigmetrics03} hold in the steady-state, they are easily violated during the transient phase. For example, $\mu_0 \leq 26.9$ for $\forall x_0$ during steady-state. However, during the transient phase, $\mu_0$ may increase abruptly by several factors. In Fig.~\ref{fig:x0flap1C10C}, for instance, $\mu_0(21.7s)=72\%$.
    \end{itemize}
    What are the lowest and highest UDP utilizations during the transient phase? Note that these extreme UDP utilizations during the transient regime depend on the measurement interval / window. With large windows, the transient behavior gets diffused and evened out by the adjacent steady-state results. The objective of this work is to obtain the \emph{extreme packet-level UDP utilizations} (i.e., the highest or lowest utilizations as measured for each packet) by analysis.

    These observations have crucial and practical implication. Internet flows, dominated by short Web transfers, often see fluctuating available bandwidth and may adapt their sending rates.  For example, in periods of high UDP link utilizations, the number of Web document transfers may dwindle and their transfer completion times may get inflated.

    Since the transient regime represents a departure from a stable queue state, its behavior may be influenced by the earlier steady state, as we shall soon see. Therefore, it is fitting to briefly discuss the steady state behaviors.

    \subsection{Background on Steady State Models}\label{subsec:backgrnd}
    As mentioned earlier, steady state CHOKe models, e.g., ~\cite{ChokeToN04,ChokeSigmetrics03}, are important for predicting transient behaviors. We call the model in~\cite{ChokeSigmetrics03} the \emph{overall loss model} and the one in~\cite{ChokeToN04} the \emph{spatial distribution model}. Both models assume the CHOKe dropping and RED-based dropping are reversed for analytic simplicity, as shown in Fig.~\ref{fig:chokeQReveresed}. The only independent parameter in both models is the UDP flow arrival rate, $x_0$. Thereupon, change in $x_0$ causes the departure from current UDP utilization and kicks start the transient regime.

        \begin{figure}[tbh!]
            \centering
            \includegraphics[width=0.45\textwidth]{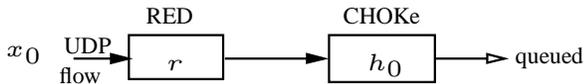}
            \caption{Schematic of dropping in CHOKe.}
            \label{fig:chokeQReveresed}
    \end{figure}

    \subsubsection{Overall Loss Model (OLM)}\label{subsub:olmModel}
    This model derives the steady-state UDP utilization $\mu_0$ and buffer share $h_0$ by first deriving the flow loss probability incurred by both CHOKe and RED parts of the CHOKe queue. The analysis results in a nonlinear numerical relation between $\mu_0$, $h_0$ and UDP input rate $x_0$  given by \eqref{eq:olm_u0}, \eqref{eq:olm_input}, and is graphically shown in Fig.~\ref{fig:olm-steadystate}. Further details can be found in \cite{ChokeSigmetrics03}.

        \begin{equation}
         \mu_0              = \frac{\ln[(1-h_0)/(1-2h_0) ]}{ [ (1-h_0)/(1-2h_0) ]+ \ln[(1-h_0)/(1-2h_0) ] }   \label{eq:olm_u0}
       \end{equation}

       \begin{equation}
         x_0(1-r)/C = \mu_0/(1-2h_0)     \label{eq:olm_input}
       \end{equation}

          \begin{figure}[h]
            \centering
            \includegraphics[width=0.4\textwidth]{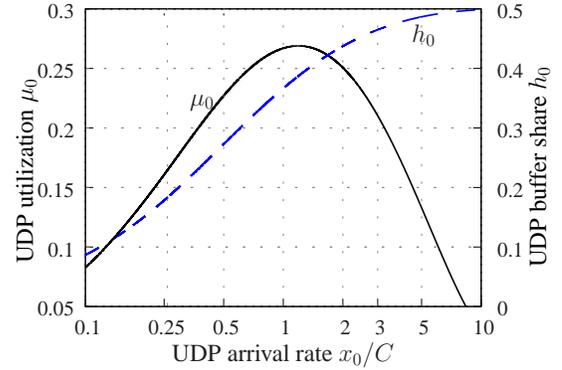}
            \caption{Steady-state UDP utilization $\mu_0$ and buffer share $h_0$.}
            \label{fig:olm-steadystate}
      \end{figure}

      The figure verifies the limit and asymptotic properties reviewed in Sec.~\ref{subsec:introChoke}. The  steady state simulation results in Fig.~\ref{fig:4x12xUtil} (areas outside the shade) also match $\mu_0$ values that can be found from Fig.~\ref{fig:olm-steadystate}.

    \subsubsection{Spatial Distribution Model (SDM)}
        To characterize the peculiar transient properties of a CHOKe queue, Tang, Wang and Low developed a novel ordinary differential equation model~\cite{ChokeToN04}. The model captures the spatial distribution of flows in a CHOKe queue using queue properties at the tail and head as boundary conditions. This spatial distribution includes the packet velocity and the probability of finding a flow packet at any position in the queue. An important concept introduced is \emph{thinning} which refers to the decaying of UDP packet velocity as the packet moves along towards the queue head. This steady state model, and the notion of thinning, are useful for studying the transient properties as well. We present the basic model, and provide pertinent assumptions that enable us to extend the model to the transient regime in Sec.~\ref{sec:modifiedSDM}.

    \section{Modeling the Transient Regime}\label{sec:modifiedSDM}
    Before delving into the transient model, we explain the transient behaviors using an argument, which we call the \emph{rate conservation law}. This argument also provides an insight to extending the SDM to transient regimes.
    \subsection{Rate Conservation Argument}\label{subsec:rateConservation}
    Before a UDP packet can be admitted into a CHOKe queue, it must survive both the RED and the CHOKe based dropping. The probability of packet admission into queue is then $(1-r)(1-h_0(t))$. Once in the queue, UDP packets can still be lost. This is because incoming packets that evade RED-based dropping (with probability $1-r$) may trigger flow matching (with probability $h_0(t)$) and cause dropping of the matched packets. In addition, UDP packets can also leave the queue due to transmission with rate $\mu_0(t)C$. Summarizing, we get a system invariant that captures the rate of change in UDP buffered packets as follows:
    \begin{align}
    \frac{db_0(t)}{dt}  &=x_0(t)(1-r)(1-h_0(t))-x_0(t)(1-r)h_0(t)-\mu_0(t)C \nonumber \\
                        & =x_0(t)(1-r)(1-2h_0(t))-\mu_0(t)C    \label{eq:invariance}
    \end{align}


    Let us call \eqref{eq:invariance} the \emph{ rate conservation law}. That is, the rate of change in UDP buffer occupancy is the difference between the flow queueing rate $x_0(t)(1-r)(1-h_0(t))$ and the outgoing rate. The outgoing rate in turn is the sum of the departure/transmission rate given by $\mu_0(t)C$ and the leaking rate given by $x_0(t)(1-r)h_0(t)$. Here the leaking rate denotes the rate with which a queued UDP packet matches the incoming packet and is consequently dropped.

   The corresponding equation for a TCP flow is, where note $h_1=0$ due to large number assumption (see \eqref{eq:tcpNoMatch}),
    \begin{equation}\label{eq:tcp_invariance}
        \frac{db_1(t)}{dt}  =  x_1(t)(1-r)- \mu_1(t)C,
    \end{equation}
    where, since the link is fully used,
    \begin{align}\label{eq:mu_1}
        \mu_1(t)    &=  \frac{1-\mu_0(t)}{N}.
    \end{align}

   Since, trivially, $b(t)=b_0(t)+Nb_1(t)$, we get
   \begin{equation}\label{eq:b(t)invariance}
        \frac{db(t)}{dt}= \frac{db_0(t)}{dt} + N\frac{db_1(t)}{dt}.
   \end{equation}

    We remark  that during a stable / steady state  $db_i(t)/dt\approx 0,~i\in(0,1)$ (see Fig.~\ref{fig:dbdt}). 

    Now, let us assume  an abrupt change in UDP arrival rate $x_0$ and note the following in the immediate aftermath:
%
   \textit{ TCP flows react slowly in response to the sudden change in UDP arrival rate. Notably, TCP flows react in an RTT timescale, but the transient behavior lasts for sub-RTT scales}.
    Hence, during the transient phase,
    \begin{equation}\label{eq:dbdtEqualsdb0dt}
        \frac{db_1(t)}{dt} \approx 0 \quad \Rightarrow \quad \frac{db(t)}{dt} \approx \frac{db_0(t)}{dt}
    \end{equation}

    \begin{figure}[bht!]\centering
                \subfloat[Experiment 1 $(0.5C/2C)$. ] {\label{fig:dbdt4x}\includegraphics[width=0.4\textwidth]{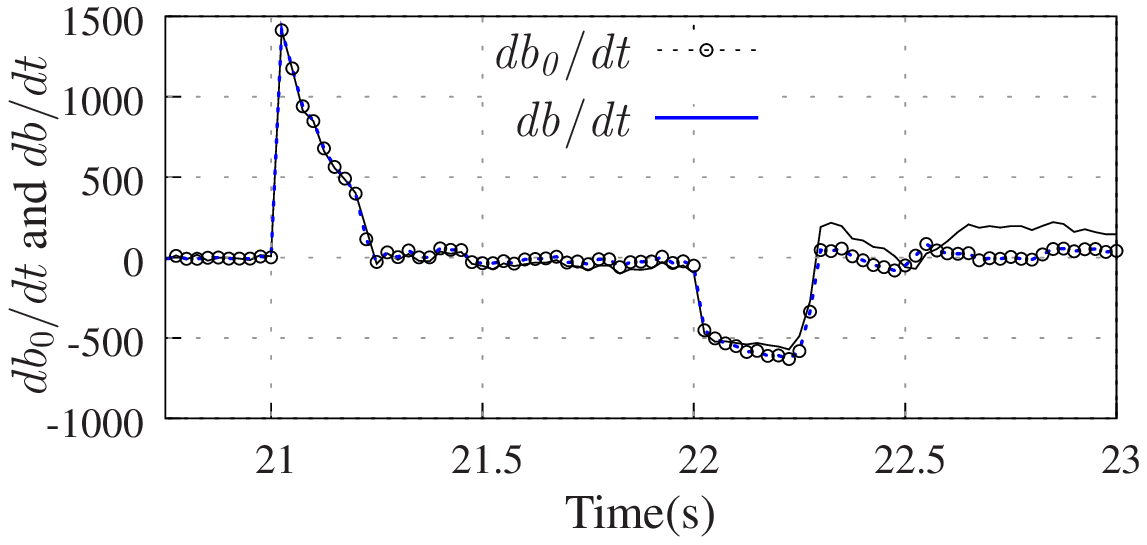}}  
                \hspace{5mm}
                \subfloat[Experiment 2 $(0.25C/3C)$.]{\label{fig:dbdt12x}\includegraphics[width=0.4\textwidth]{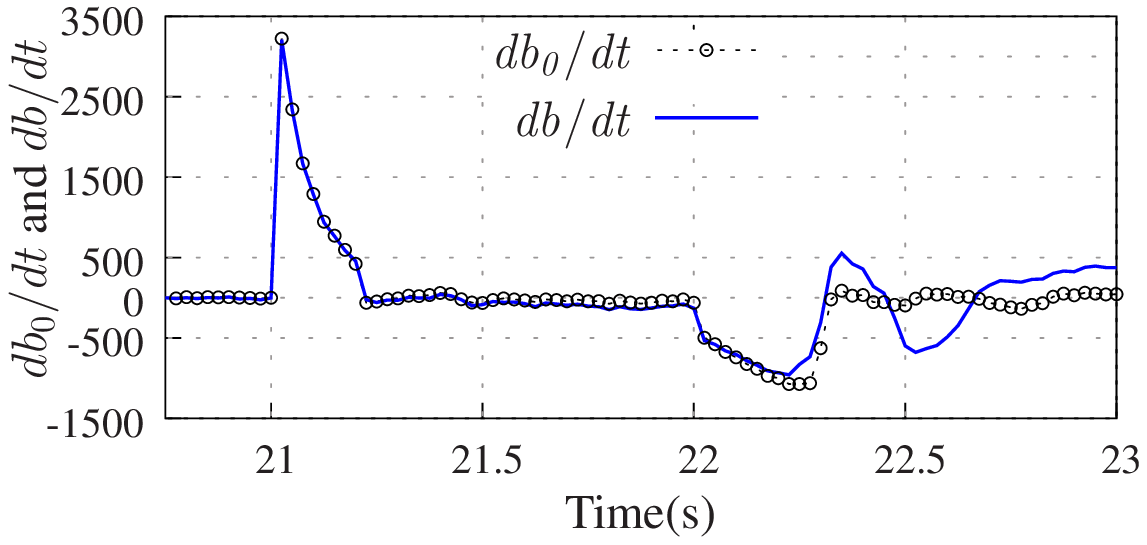}}
            \caption{Transient regime: $db/dt \approx db_0/dt$ and $db_1/dt \approx 0$. Steady state: $db/dt\approx 0 $, $db_0/dt\approx 0 $, $db_1/dt\approx 0 $.}
            \label{fig:dbdt}
    \end{figure}

    Fig.~\ref{fig:dbdt} plots the simulation results for the two experiments of Sec.~\ref{subsec:motivation}. The figures confirm the argument above and Eq.~\eqref{eq:dbdtEqualsdb0dt}. As can be seen, for most parts of the simulations, the rate of change in total buffer occupancy $b$ is almost solely due to change in UDP buffer occupancy $b_0$.

    We warn the reader that \eqref{eq:dbdtEqualsdb0dt} may not hold outside the transient regime. For instance, when the UDP flow arrival rate plummets at $t=22$s, TCP flows respond by increasing their sending rates after around a round-trip delay (see especially Fig.~\ref{fig:dbdt12x}). Subsequently, $db_1/dt \neq 0$, but $db_0/dt\approx 0$. From \eqref{eq:b(t)invariance}, $db/dt\approx N db_1/dt$. That is, $db/dt$ swings from $db_0/dt$ during the transient phase to $Ndb_1/dt$ following the transient phase. However, since the TCP flows are largely in a congestion avoidance phase, the rise in $Ndb_1/dt$ is not as significant as that of $db_0/dt$ in the transient phase. After absorbing TCP bursts for a while (in a few round-trip cycles), the queue eventually settles to a new steady state determined by the new UDP arrival rate, and then $db_i/dt \approx 0,~i\in(0,1) $ once again.

    {\bf Remark:} Combining  \eqref{eq:tcp_invariance} and \eqref{eq:dbdtEqualsdb0dt}, we conclude that TCP packet arrival rate to the CHOKe queue matches its transmission rate during the transient regime.

    Now we are in a position to explain the transient behaviors shaded in Fig.~\ref{fig:4x12xUtil} and Fig.~\ref{fig:x0flap1C10C}  from the perspective of the rate conservation law. Rearranging \eqref{eq:invariance},  we get 
        \begin{equation}\label{eq:invariance3}
            \mu_0(t) = \frac{x_0(t)(1-r)}{C} ( 1-2h_0(t)) -\frac{1}{C}\frac{db_0(t)}{dt}
        \end{equation}

    Interestingly, \eqref{eq:invariance3} captures all pertinent behaviors of the system. In the steady state, $db_0(t)/dt\approx0$, and

        \begin{equation}\label{eq:invariance4}
            \mu_0 = \frac{x_0(1-r)}{C} ( 1-2h_0)
        \end{equation}
 which is \eqref{eq:olm_input}. \eqref{eq:invariance4} is a key equation in the OLM  model and can reproduce the steady state results depicted in Fig.~\ref{fig:olm-steadystate}.

For the transient behaviors, we explain only the dips in $\mu_0$ shown in Fig.~\ref{fig:4x12xUtil} but similar arguments follow for the peaks as well. An abrupt injection of UDP rate $x_0$ at $t=21$s rapidly ramps up the second term on the right hand side (r.h.s.) of \eqref{eq:invariance3} (see also Fig.~\ref{fig:dbdt}). Despite the rapid rise of $x_0(1-r)/C$, its contribution to $\mu_0(t)$ is counteracted by a corresponding rise of $h_0(t)=b_0(t)/b(t)$ (see Fig.~\ref{fig:12BufferChange} where $h_0(t)\rightarrow 45 \%$). Therefore, $\mu_0(t)$ is mainly  influenced inversely by the rate $db_0(t)/dt$ (cf. Fig.~\ref{fig:dbdt} and Fig. \ref{fig:4x12xUtil}). Specifically, focusing on Fig.~\ref{fig:4x12xUtil}  and using measurement intervals of 10ms, we observed the following : when $x_0$ increases 12-fold at $t=21$s, $db_0/dt \rightarrow  3500$ (as shown in Fig. \ref{fig:dbdt12x}) but the utilization $\mu_0 \rightarrow 3.75\%$. Conversely, when $x_0$ slackens by a factor of 12 at $t=22$s, the $db_0/dt$ falls to 1000 below $0$ (as shown in Fig. \ref{fig:dbdt12x}) but the UDP link bandwidth share soars to 56.5\%.

            \begin{figure}[tbh!]
            \centering
            \includegraphics[width=0.45\textwidth]{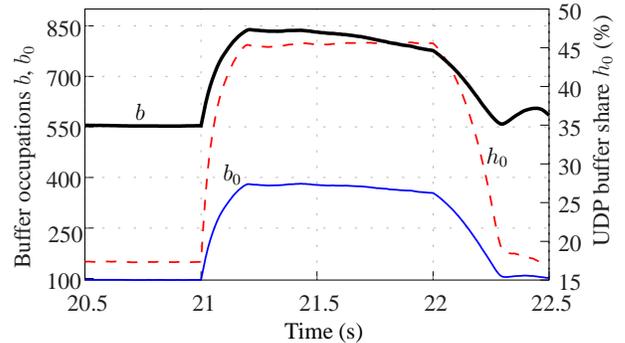}
            \caption{At $t=21$, UDP buffer share $h_0$ jumps radically, nullifying the impact on $\mu_0$ of the sudden change in  $x_0$.}
            \label{fig:12BufferChange}
    \end{figure}

While the rate conservation argument explains the transient queue dynamics well, due to several dynamically changing parameters ($h_0(t)$, $b_0(t)$, $b(t)$) in the transient regime, it is difficult to derive quantitative $\mu_0$ results directly from \eqref{eq:invariance3}. Nevertheless, some of the insights we gained become useful for extending the SDM model to the transient regime, as we shall soon see.

    \subsection{Modified Spatial Distribution Model}\label{subsec:modifiedSDM}
    This section is dedicated to the development of the SDM in concert with the transient regime. A schematic diagram of this model is illustrated in Fig. \ref{fig:udpVelocityDecay}.

    \begin{figure}[h!]
            \centering
            \includegraphics[width=0.5\textwidth]{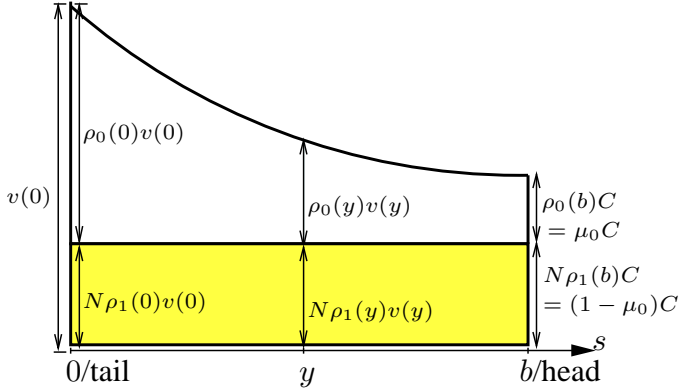}
            \caption{Schematic diagram: Decay of UDP velocity $\rho_0(y)v(y)$ in queue. Look at the similarity to Fig.~\ref{fig:vyVSy} with $x_0=10C$.}
            \label{fig:udpVelocityDecay}
    \end{figure}

    \subsubsection{Model parameters}
SDM can be described by a few key parameters (see Fig.~\ref{fig:udpVelocityDecay} and Table~\ref{tab:notation}). The parameters are the queue position/point/slot $y\in[0,b]$, the packet velocity $v(y)$ at $y$, the probability $\rho_i(y)$ of finding a flow $i$ packet at $y$, and the queueing delay $\tau(y)$ for the packet at the tail to arrive at slot $y$. Queue position $y$ is indexed from tail to head as $\{0,\cdots,~b\}$. The packet velocity $v(y)$ is the speed with which packets move towards the head of the queue, and is defined as:
    \begin{equation}\label{eq:velocity}
        v(y)=dy/dt.
    \end{equation}

    The packet velocity at queue tail $v(0)$ is simply the full queueing rate $\sum_{i=0}^N x_i(1-r)(1-h_i)$ (see Fig.~\ref{fig:chokeQReveresed}). At queue head, however, $v(y)$ is merely the link capacity, i.e., $v(b)=C$. The packet velocity $v(y)$ is related to the queueing delay $\tau(y)$ accumulated in going from tail $y=0$ to slot $y$ as follows,
        \begin{equation}\label{eq:qDelay}
        dt=dy/v(y) \qquad \Rightarrow \qquad \tau(y)= \int\limits_0^y \frac{1}{v(s)}ds
    \end{equation}
where the equation on the left side is obtained from \eqref{eq:velocity}.

    Of course, the full queueing delay is $\tau(b)$.  Alternatively, it can be derived using queueing principles.  The rate of departure of TCP packets is $C(1-\mu_0)$, and average number of TCP packets in queue is given by $b(1-h_0)$. As far as the TCP flows are concerned, the model is a non-leaky queue (since $h_1=0$). Therefore, we apply Little's law to obtain,
    \begin{equation}\label{eq:fullDelay}
        \tau(b)= \frac{b(1-h_0)}{C(1-\mu_0)}.
    \end{equation}

    Another useful spatial parameter is $\rho_i(y)$---the probability of finding a flow $i$ packet at slot $y$. Trivially,
    \begin{align}
        \rho_0(y)+N\rho_1(y)& =1 \qquad y\in[0,b] \label{eq:initialConditions1}
    \end{align}

    Here, $\rho_i(y)$ is closely related to the packet velocity $v(y)$: It quantifies the fraction of flow $i$'s packet velocity at $y$ to the total packet velocity $v(y)$. For instance, at queue head $y=b$, $\rho_i(b)=\mu_i$, i.e., the probability is simply the flow utilization.

    Summarizing the two important parameters $\rho_i(y)$ and $v(y)$ at queue tail and head, the following boundary conditions apply, which are also illustrated in Fig.~\ref{fig:udpVelocityDecay}. Here, we ignore TCP flow matching as discussed earlier for \eqref{eq:tcpNoMatch}, i.e., $h_1 \approx 0$.

        \begin{align}
        v(0)& = x_0(1-r)(1-h_0)+Nx_1(1-r)(1-h_1) \nonumber \\
            &\approx  x_0(1-r)(1-h_0)+Nx_1(1-r) \label{eq:icv(y)1} \\
        \rho_0(0) &=  \frac{x_0(1-h_0)(1-r)}{[x_0(1-h_0)+Nx_1](1-r)} \approx \frac{x_0(1-h_0)}{x_0(1-h_0)+Nx_1} \label{eq:rho_0(0)} \\
        \rho_1(0) &= \frac{x_1}{x_0(1-h_0)+Nx_1}
    \end{align}
    \begin{equation}\label{eq:icAtb}
        v(b)=C, \qquad \rho_0(b) =\mu_0, \qquad \rho_1(b)= \mu_1=\frac{1-\mu_0}{N}
    \end{equation}

       Note that $\rho_1(b)$ in \eqref{eq:icAtb} is the same as \eqref{eq:mu_1}.  As noted in Sec.~\ref{subsec:rateConservation}, TCP transmission rates do not change during the transient regime, resulting in constant TCP packet velocities throughout the queue. In a congested CHOKe queue with high UDP rate $x_0$, however, the total packet velocity $v(y)$ is continuously decreasing because UDP arrivals trigger packet drops through flow matching. In fluid terms, we say the UDP fluid gets \emph{thinned} as it moves along the queue. We formalize the notion of thinning and use it to derive the slot parameters next.

       \subsubsection{Ordinary differential equation model}\label{subsec:ode}
      The UDP portion of the packet velocity at the tail ($y=0$) is given by $\rho_0(0)v(0)=x_0(1-h_0)(1-r)$ and the amount of UDP fluid in small time $dt$ at the tail by $\rho_0(0)v(0)dt$. The corresponding values at $y$ are $\rho_0(y)v(y)$ and $\rho_0(y)v(y)dt$, respectively. Traveling from queue tail to $y$ takes $\tau(y)$ during which time $x_0(1-r)\tau(y)$ new packets would arrive to the queue. Each arrival triggers a flow matching trial and drops the small volume of fluid with success probability $1/b$. The probability that the UDP volume escapes matchings by all arrivals  is $(1-1/b)^{x_0(1-r)\tau(y)}$. The UDP packet velocity at $y$ is therefore thinned or weakened as,
    \begin{align}
        \rho_0(y)v(y)&=\rho_0(0)v(0)(1-1/b)^{x_0(1-r)\tau(y)} \label{eq:chokeSteadyThinning}
    \end{align}
    On the other hand, there is no thinning for TCP and hence a constant  TCP packet velocity throughout. 
    \begin{align}
            N\rho_1(y)v(y)&= N\rho_1(0)v(0)=(1-\mu_0)C     \label{eq:tcpThroughput}
    \end{align}

    During the transient regime, \eqref{eq:tcpThroughput} is still valid due to the slow reaction of TCP congestion control, as discussed in Sec.~\ref{subsec:rateConservation}.

    Rearranging \eqref{eq:tcpThroughput} and using \eqref{eq:initialConditions1}, we get
    \begin{align}
            v(y)&= \frac{\rho_1(0)v(0)}{\rho_1(y)}=\frac{(1-\mu_0)C}{N \rho_1(y)} = \frac{(1-\mu_0)C}{1-\rho_0(y)}    \label{eq:v(y)}
    \end{align}


Define parameters $a$ and $\beta$ as follows, which will often be used throughout the rest of the paper:
    \begin{align}
      a: &=\frac{1-\rho_0(0)}{\rho_0(0)} \nonumber \\
      \beta: &=\ln(1-1/b) \nonumber
    \end{align}

    Using \eqref{eq:initialConditions1}, \eqref{eq:icv(y)1}, \eqref{eq:rho_0(0)} and \eqref{eq:tcpThroughput}, $a$ can equivalently be rewritten as,
    \begin{equation}\nonumber
        a=\frac{1-\rho_0(0)}{\rho_0(0)} = \frac{N\rho_1(0)}{\rho_0(0)}=\frac{N\rho_1(0)}{\rho_0(0)}\frac{v(0)}{v(0)}=\frac{(1-\mu_0)C}{x_0(1-r)(1-h_0)}.
    \end{equation}

    Taking logarithm of key equation~\eqref{eq:chokeSteadyThinning} first and then differentiation w.r.t. $y$, we get
    \begin{align}
        \ln(\rho_0(y)v(y))&=\ln(\rho_0(0)v(0))+ x_0(1-r)\beta \tau(y) \nonumber \\
        \frac{\rho_0^{'}(y)}{\rho_0(y)}+ \frac{v^{'}(y)}{v(y)} &= 0+ \frac{x_0(1-r) \beta}{v(y)} \label{eq:lnThinningv1}
    \end{align}
    where $\tau{'}(y)^=1/v(y)$ from \eqref{eq:qDelay}.


    Applying \eqref{eq:v(y)} for $v(y)$ and $v^{'}(y)$ and inserting into l.h.s. of \eqref{eq:lnThinningv1}, the following ordinary differential equation (ODE) is obtained.
    \begin{align}
        \frac{\rho_0^{'}(y)}{\rho_0(y)}+ \frac{\rho_0^{'}(y)}{1-\rho_0(y)} = \frac{x_0(1-r) \beta}{v(y)} \label{eq:lnThinningv2}
    \end{align}

Eq. \eqref{eq:lnThinningv2} establishes a foundation for the analysis in the remaining part of the paper.

First, solving $\rho_0(y)$ from ODE \eqref{eq:lnThinningv2} (see Appendix for the proof), we have the following relation between $\rho_0(y)$ and $\tau(y)$:

\begin{lemma}\label{lemma-rho0(y)}
    \begin{equation}
        \rho_0(y)=\frac{e^{x_0(1-r)\beta\tau(y)}}{a+e^{x_0(1-r)\beta\tau(y)}}. \label{eq:rho_0y}
    \end{equation}
\end{lemma}
\vspace{2mm}

    In addition, substituting \eqref{eq:v(y)} for $v(y)$ in \eqref{eq:lnThinningv2} gives,
    \begin{equation}
        \frac{\rho_0^{'}(y)}{\rho_0(y)}+ \frac{\rho_0^{'}(y)}{1-\rho_0(y)} = \frac{x_0(1-r) \beta}{(1-\mu_0)C}(1-\rho_0(y)). \nonumber
       \end{equation}
    Dividing by $(1-\rho_0(y))$, we obtain
       \begin{equation}
         \frac{\rho_0^{'}(y)}{\rho_0(y)(1-\rho_0(y))}+ \frac{\rho_0^{'}(y)}{(1-\rho_0(y))^2} = \frac{x_0(1-r) \beta}{(1-\mu_0)C}   \nonumber
       \end{equation}
     Upon integrating w.r.t. $y$, we get
       \begin{align}
         \int\limits_0^y \frac{x_0(1-r) \beta}{(1-\mu_0)C}ds &= \int\limits_0^y \frac{\rho_0^{'}(s)}{\rho_0(s)(1-\rho_0(s))}ds+ \int\limits_0^y \frac{\rho_0^{'}(s)}{(1-\rho_0(s))^2}ds \nonumber \\
         \frac{x_0(1-r) \beta}{(1-\mu_0)C}s\bigg|_0^y    &=\ln\bigg[\frac{\rho_0(s)}{1-\rho_0(s)}\bigg] \bigg|_0^y +\frac{1}{1-\rho_0(s)}\bigg|_0^y \nonumber
        \end{align}
    Let $K=\frac{x_0 (1-r)\beta}{(1-\mu_0)C}$. We then get:
        \begin{align}
            Ky&= \ln\bigg[\frac{\rho_0(y)}{1-\rho_0(y)}\frac{1-\rho_0(0)}{\rho_0(0)}\bigg]+\frac{1}{1-\rho_0(y)}-\frac{1}{1-\rho_0(0)} 
        \end{align}
from which, the following relation between $y$ and $\rho_o(y)$ is easily established:

\begin{lemma}\label{lemma-y}
    \begin{equation}
y = \frac{1}{K}\ln\bigg[\frac{a\rho_0(y)}{1-\rho_0(y)}\bigg]+\frac{1}{K}\frac{\rho_0(y)-\rho_0(0)}{(1-\rho_0(y))(1-\rho_0(0))}\label{eq:qPositionYv1}
    \end{equation}
where $K=\frac{x_0 (1-r)\beta}{(1-\mu_0)C}$ and as above, $a:=\frac{1-\rho_0(0)}{\rho_0(0)}=\frac{(1-\mu_0)C}{x_0(1-r)(1-h_0)}$.
\end{lemma}

{\bf Remark:} Lemma \ref{lemma-rho0(y)} and Lemma \ref{lemma-y} are crucial, through which, the slot and its associated parameters are inter-related as defined by $\rho_0(y)$ \eqref{eq:rho_0y},  $y$ \eqref{eq:qPositionYv1}, $v(y)$ \eqref{eq:v(y)} and $\tau(y)$ \eqref{eq:lnThinningv3}. However, from the model in~\cite{ChokeToN04}, it is difficult to obtain for each slot $y$ corresponding explicit expressions of the spatial parameters $\rho_0(y)$, $v(y)$ and $\tau(y)$, hence limiting the application of results in \cite{ChokeToN04} to the transient analysis. This is because $\rho_0(y)$ and $v(y)$ were given in terms of queueing delay $\tau(y)$, but no explicit relations were afforded between queue slot $y$ and anyone of the triplets $(v(y), \rho_0(y), \tau(y))$ in \cite{ChokeToN04}.

{\bf Remark:} More importantly, the analytical principles leading to Lemma \ref{lemma-rho0(y)} and Lemma \ref{lemma-y} also apply to the transient regime, making the later transient analysis possible.

        %


    \begin{figure*}
        \begin{minipage}[htb!]{0.98\linewidth}\centering
            \subfloat[UDP packet probability vs queue slot.] {
            \label{fig:rhoyVSy}\includegraphics[width=0.3\textwidth]{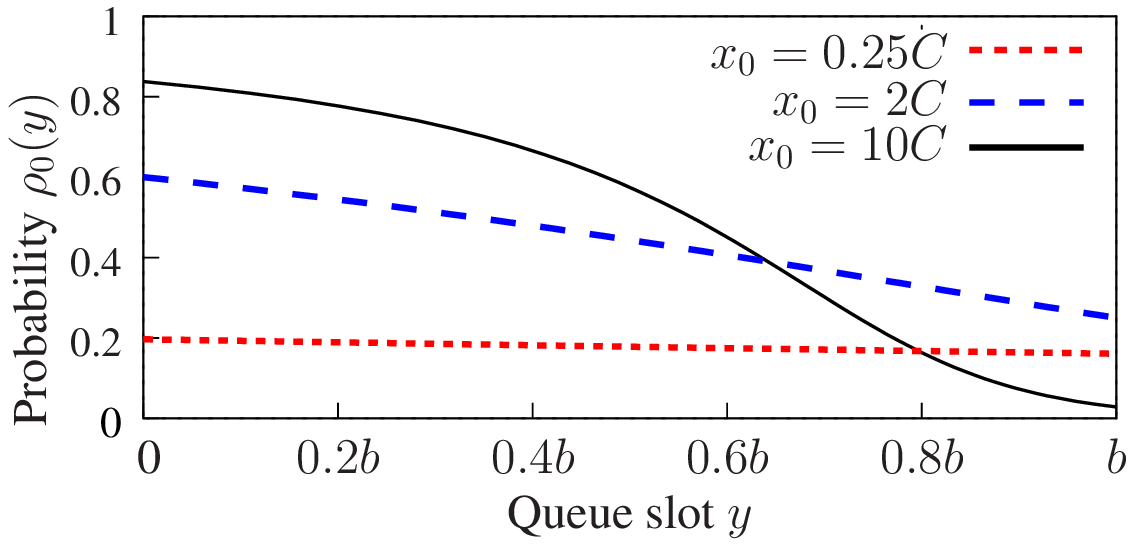} }
            \subfloat[Packet velocity at queue slot.] {
            \label{fig:vyVSy}\includegraphics[width=0.3\textwidth]{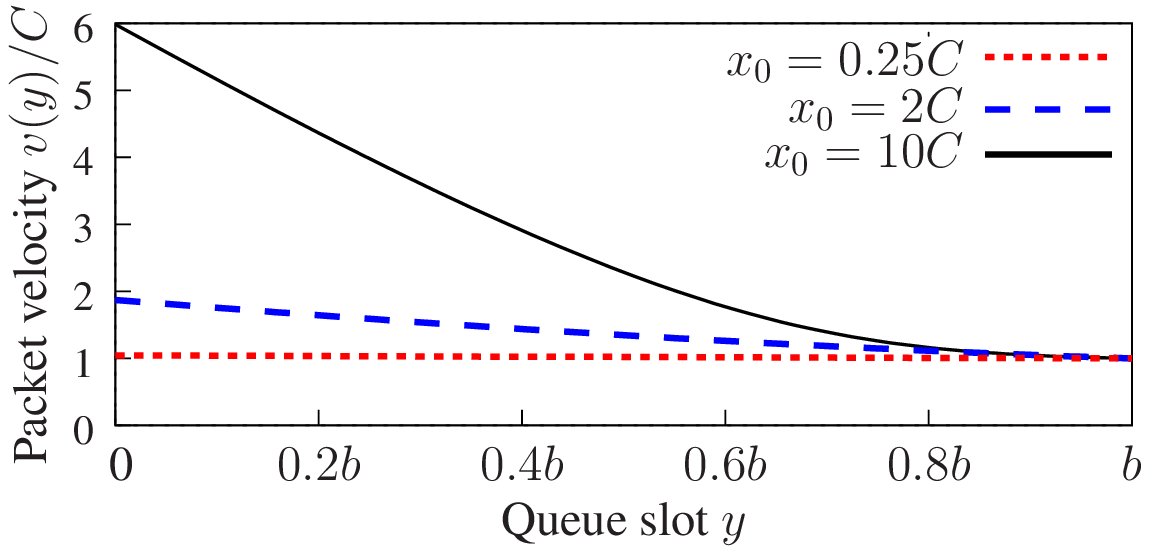}}
            \subfloat[Queuing delay to reach at a queue slot.] {
            \label{fig:tauyVSy}\includegraphics[width=0.3\textwidth]{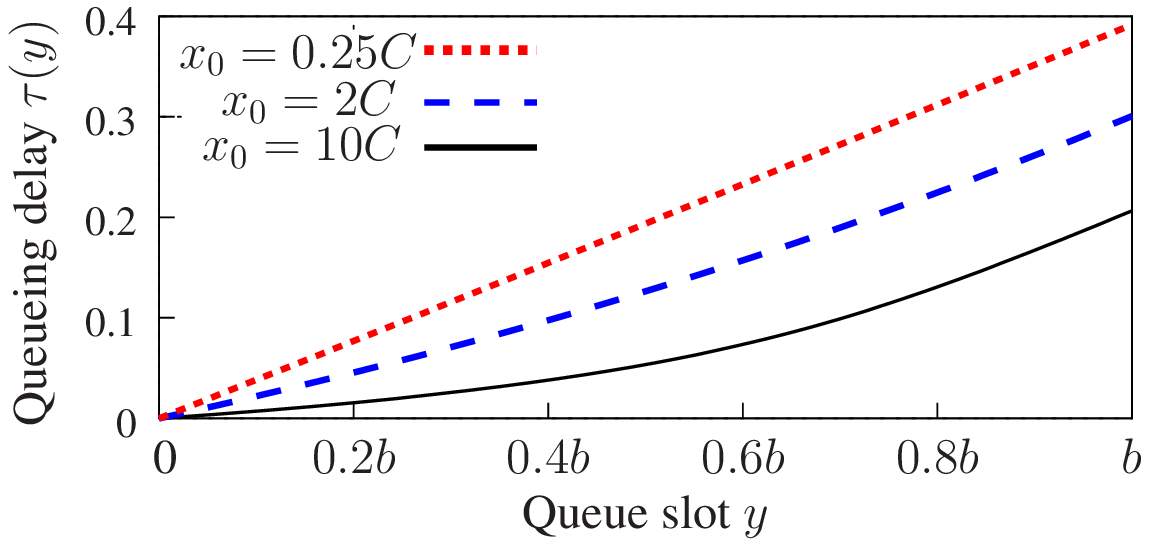}}
            \caption{Spatial characteristics of a CHOKe queue under different intensities of input UDP arrival rates.}
            \label{fig:ss_CHOKEproperties}
        \end{minipage}
    \end{figure*}

    \subsubsection{Properties of queue dynamics}\label{subsub:sdmProperties}
    This section briefly outlines properties of queue dynamics $\rho_0(y)$, $v(y)$ and $\tau(y)$. The detailed proofs are in the Appendix.

We start with $\rho_0(y)$. To prove the properties of $\rho_0(y)$ in Lemma \ref{eq:lemmaRho_0Y}, we need the following intermediate result.
\begin{lemma}\label{lemma:intermediateR}
    \begin{equation}\nonumber
        1 \leq \frac{1-\mu_0}{1-h_0} \leq 2
    \end{equation}
\end{lemma}

The following properties hold for $\rho_0(y)$, and follow from \eqref{eq:rho_0y} and Lemma \ref{lemma:intermediateR}.
    \begin{lemma}\label{eq:lemmaRho_0Y}
       Given $x_0>0$, 
        \begin{enumerate}[(a)]
             \item If $x_0 \leqslant C/2$, $\rho_0(y)$ is strictly convex decreasing.
             \item Otherwise, $\rho_0(y)$ is concave decreasing from $\rho_0(0)$ to $\rho_0^*$, and  convex decreasing from $\rho_0^*$ to $\rho_0(b)$ where the critical point $(y^*,~\rho_0^*)$ is given by
                    \begin{align}
                        \rho_0^*    &=\rho_0(y^*)= \frac{1}{3} \nonumber \\
                        y^*         &= \frac{1}{K} \ln\bigg(\frac{a}{2}\bigg)+\frac{1}{K} \frac{1-3\rho_0(0)}{2(1-\rho_0(0))}. \nonumber
                \end{align}
        \end{enumerate}
    \end{lemma}

For $v(y)$, applying \eqref{eq:rho_0y} to \eqref{eq:v(y)} first, then the following property can be verified from \eqref{eq:v(y)}
    \begin{lemma}\label{eq:lemmavY}
        Given $x_0>0$, and boundary values $v(0)$ and $v(b)$ defined respectively by \eqref{eq:icv(y)1} and \eqref{eq:icAtb}, the packet velocity $v(y)$ is convex decreasing throughout the queue.
    \end{lemma}

For $\tau(y)$, the following property holds.
    \begin{lemma}\label{eq:lemmaTauY}
        Given $x_0>0$, $\tau(0)=0$ and $\tau(b)$ defined by \eqref{eq:fullDelay},  queueing delay $\tau(y)$ is strictly convex increasing.
    \end{lemma}

        To visualize the queue dynamics, Fig.~\ref{fig:ss_CHOKEproperties} plots for each queue position $y$ the probabilities $\rho_0(y)$, the packet velocities $v(y)$  and the queueing delays $\tau(y)$. For low and moderate rate $x_0$, the spatial properties  $v(y)$ and $\rho_0(y)$ look uniform in queue, and the queueing delay $\tau(y)$ is approximately linearly rising, like in a regular non-leaky queue. With increasing rate $x_0$, however, the spatial distribution becomes increasingly asymmetrical in queue, with most of the UDP packets piled up closer to the tail and the packet velocity sharply decreasing to $C$ as we move towards the head of queue. In such cases, since the packet velocities are nonuniform in queue, the queueing delays are also nonuniform. When $x_0\rightarrow \infty$,  the queueing delay becomes $\frac{1}{2}\frac{b}{C}$---half of the queueing delay possible in a non-leaky queue with the same backlog size $b$.

    \subsection{Analysis on the Transient Behavior}\label{subsec:transient}
    The transient behavior is a transition between two stable queue states excited by a change in UDP arrival rate. Without loss of generality, let us assume that the UDP arrival rate changes from $x_0$ to $x_{02}$ at time $t=0$. At that instant, the snapshot of the queue exhibits the steady state characteristics defined by $x_0$ (see Fig.~\ref{fig:udpVelocityDecay}). We are interested in transient behaviors that surface within a time of $\tau(b)$ (the full queueing delay) after rate change. Thereafter, we assume the queue enters the steady state defined by new rate $x_{02}$. We shall focus on the transition regime in $[0,\tau(b)]$.

While Lemma \ref{lemma-rho0(y)} and Lemma \ref{lemma-y} lay a foundation for the analysis, we still need to make assumptions to realize it. Specifically, we {\em boldly} assume, during $[0,\tau(b)]$,
    \begin{enumerate}[I: ]
        \item the buffer occupancy remains constant at $b$, and \label{tran:AI}
        \item TCP arrival rates remain the same.   \label{tran:AII}
    \end{enumerate}

    {\bf Remark:} It is worth highlighting that while Assumption \ref{tran:AII} is valid or accurate based on our observation of slow TCP reactions in Sec.~\ref{subsec:rateConservation}, Assumption \ref{tran:AI} may be very conservative (see Fig.~\ref{fig:12BufferChange}). Incidentally, this assumption may be the cause of approximation errors in the analysis. Nevertheless, results based on these assumptions are highly satisfactory.


\vspace{5mm}

   We are now ready to conduct the analysis. Consider at time $t=0$ the UDP fluid at position $y$. It has a velocity $\rho_0(y)v(y)$, and a remaining age in queue of $\tau(b)-\tau(y)$. Owing to FIFO, the UDP flow transmission rate at time $\tau(b)-\tau(y)$ is due to the transmission of this UDP fluid. By transmission time $\tau(b)-\tau(y)$, the UDP velocity has been thinned according to the new rate $x_{02}$ as:
    \begin{equation}\label{eq:newThinnedv_0}
        \rho_0(y)v(y)(1-1/b)^{x_{02}(1-r)(\tau(b)-\tau(y))}.
    \end{equation}

{\bf Remark:} \eqref{eq:newThinnedv_0} is similar to \eqref{eq:chokeSteadyThinning}. Unlike \eqref{eq:chokeSteadyThinning}, however, there are two distinct parts of thinning for the UDP fluid located at slot $y\in[0,b]$ at $t=0$. First, in going from tail to $y$, the thinning is according to $x_0$, and the duration of thinning is the queueing delay so far, i.e., $\tau(y)$.  This part is reflected in \eqref{eq:newThinnedv_0}  by $\rho_0(y)v(y)$. Second, for the remaining duration $\tau(b)-\tau(y)$, the thinning is due to the new rate $x_{02}$. Overall thinning, before eventual transmission, of the UDP fluid found at slot $y$  at $t=0$ is then,
\begin{equation}\nonumber
 \rho_0(0)v(0)(1-1/b)^{x_0(1-r) \tau(y)}(1-1/b)^{x_{02} (1-r)(\tau(b)-\tau(y))}.
\end{equation}

Above, we use the same RED / congestion-based drop probability $r$ even when the UDP arrival rate changes. In fact, our extensive simulations show that  $r$ is often insignificant and can be ignored altogether. CHOKe's excessive flow dropping keeps the average queue size $\mathit{avg}$ in check, and this in turn lowers the $r$ in comparison to that in plain RED. Hereafter we choose to ignore $r$. That is,
     \begin{align}\label{eq:neglectRED}
        r &\approx 0.
    \end{align}
Here, we would like to remark that similar observation or assumption has been made for CHOKe analysis in the literature~\cite{ChokeSigmetrics03}.

    With Assumption \ref{tran:AII}, the total packet velocity $v$ at time $\tau(b)-\tau(y)$ is the sum:
    \begin{align} \label{eq:newThinnedvy}
                v(\tau(b)-\tau(y))&=\rho_0(y)v(y)(1-1/b)^{x_{02}(\tau(b)-\tau(y))} \nonumber \\
                                        & \qquad +(1-\mu_0)C
    \end{align}
where the second term represents the velocity contributed by the TCP flows, given by \eqref{eq:tcpThroughput}.

The instantaneous UDP link utilization $\mu_0$ follows from \eqref{eq:newThinnedv_0} and \eqref{eq:newThinnedvy} simply as,
    \begin{equation}
        \mu_0(\tau(b)-\tau(y)) = \frac{\rho_0(y)v(y)(1-1/b)^{x_{02}(\tau(b)-\tau(y))}}{v(\tau(b)-\tau(y))}
    \end{equation}

Note that, with \eqref{eq:v(y)}, which holds for both the steady-state and the transient regime, we can express the UDP packet velocity, $\rho_0(y)v(y)$,  as
    \begin{align}
            \rho_0(y)v(y)&= \frac{\rho_0(y)}{1-\rho_0(y)}(1-\mu_0) C    \label{eq:v_0(y)}
    \end{align}
with which, we further obtain
    \begin{subequations}
    \begin{align}
         \mu_0(\tau(b)-\tau(y)) &= \bigg[1+\frac{1-\rho_0(y)}{\rho_0(y)}\bigg(1-\frac{1}{b}\bigg)^{-x_{02}(\tau(b)-\tau(y))}\bigg]^{-1}  \nonumber \\
                                   &= \bigg[1+\frac{1-\rho_0(y)}{\rho_0(y)}e^{-x_{02}\beta(\tau(b)-\tau(y))}\bigg]^{-1} \label{eq:newThinnedu_01}\\
                                   &= \bigg[1+a e^{-x_{02}\beta\tau(b)+\beta \tau(y)(x_{02}-x_0)}\bigg]^{-1}    \label{eq:newThinnedu_02}
    \end{align}
    \end{subequations}
    In obtaining the reduced forms \eqref{eq:newThinnedu_01},\eqref{eq:newThinnedu_02}, we used \eqref{eq:v_0(y)} and \eqref{eq:rho_0y} respectively. 

 \eqref{eq:newThinnedu_02} captures the evolution of UDP utilization during the transient regime. We summarize it in the following lemma.

 \begin{lemma}\label{eq:udpTranUtil}
 Assume at $t=0$, UDP arrival rate changes from $x_0$ to $x_{02}$. The UDP link utilization at time $\Delta T \in[0,\tau(b)]$  is given by, 
     \begin{align}\nonumber
        \mu_0(\Delta T)&=\bigg[1+ae^{-x_{02}\beta \tau(b)+\beta (\tau(b)-\Delta T)(x_{02}-x_0)} \bigg]^{-1}.
     \end{align}
     where, $b$ is backlog size at $t=0$ and $\tau(b)$ is given by \eqref{eq:fullDelay}.
 \end{lemma}

    It is trivial to see that $\mu_0(0)=\mu_0$. In addition, it is easy to prove that when $x_{02}=x_0$, $\mu_0(\Delta T)=\mu_0$ for $\forall \Delta T \in [0,\tau(b)]$. The proof is similar to the proof of Lemma \ref{lemma:alpha1}.

    \newtheorem{theorem}{Theorem}
    \begin{theorem}\label{eq:Thorem_maxu(t)_0}
       Assume a CHOKe queue characterized by steady state UDP probabilities $\rho_0(y)$, $y\in[0,b]$ and input UDP rate $x_0$. Further assume the UDP rate changes to $x_{02}\geq 0$ at $t=0$.
        \begin{enumerate}[(a)]
          \item The transient UDP utilization is upper bounded by $\rho_0(0)$.
          \item This upper bound can be achieved when $x_{02}=0$ and at time $\tau(b)$.
        \end{enumerate}
    \end{theorem}

    \begin{proof}
      Since $\tau(b)\geqslant \tau(y)$ (see Lemma~\ref{eq:lemmaTauY}) and $\beta <0$ in \eqref{eq:newThinnedu_01}, we have $\exp(-x_{02}\beta(\tau(b)-\tau(y)))
      \geqslant \exp(0)=1$ and
      \begin{align}\label{eq:u(t)Simplified}
        \mu_0(\tau(b)-\tau(y)) & \leqslant \bigg[1+\frac{1-\rho_0(y)}{\rho_0(y)} \bigg]^{-1} \nonumber \\
                                   &\leqslant  \rho_0(y) \leqslant \rho_0(0)
      \end{align}

      In \eqref{eq:u(t)Simplified}, we used the property that $\rho_0(y)$ is a decreasing function (see Lemma~\ref{eq:lemmaRho_0Y}) to state that $\mu_0(\tau(b)-\tau(y))\leqslant \rho_0(0)$.
      When $x_{02}=0$ in \eqref{eq:newThinnedu_01}, $\mu_0(\tau(b)-\tau(y)) = \rho_0(y)$. See Fig.~\ref{fig:relation-rho0(y)-udpTran} for the relationship between $\rho_0(y)$ in queue and transient UDP rate.
    \end{proof}

    \begin{figure}[h!]
            \centering
            \includegraphics[width=0.45\textwidth]{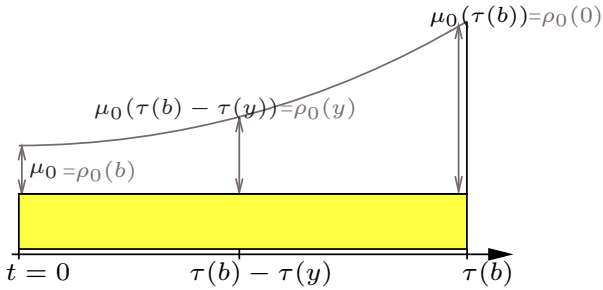}
            \caption{Relationship between steady-state $\rho_0(y)$ and transient UDP utilization $\mu_0(t)$ when the flow stops, i.e., $x_{02}=0$.}
            \label{fig:relation-rho0(y)-udpTran}
    \end{figure}

    Theorem \ref{eq:Thorem_maxu(t)_0} states that if the UDP flow stops ($x_{02}=0$), it will attain exactly the utilizations $\rho_0(y)$ in reverse order of time shown in Fig.~\ref{fig:udpVelocityDecay}. That means, the steady-state probabilities $\rho_0(y)$ associated with the previous UDP input rate $x_0$ successively turn out as transient utilizations when the flow stops. Therefore, the transient utilizations increase from $\rho_0(b)=\mu_0$ at $t=0$ to $\rho_0(0)$ at $t=\tau(b)$.

    Theorem ~\ref{eq:Thorem_maxu(t)_0} vindicates our choice to leverage steady-state CHOKe model for understanding the transient behavior.

     So far we have discussed the two special cases when $x_{02}=x_0$ (no change) and $x_{02}=0$. Now consider a scenario where  $x_{02}\notin \{0,x_0\}$, in particular $x_{02}\rightarrow \infty$. It is plausible that when $x_{02}$ is high, the exponential terms in \eqref{eq:newThinnedu_01}, or equivalently in Lemma~\ref{eq:udpTranUtil}, may become so large that the utilization may quickly plunge to very low values. From our observation, we remark that the \emph{UDP utilization in the transient phase moves in the opposite direction to the change of UDP input rate that triggers the phase. } Since the UDP rate change impacts the rate $db_0/dt$, the above observation is coincidentally the same as the one noted in Sec.~\ref{subsec:rateConservation}.

     We now generalize our findings and obtain the extreme values. Note that from Lemma~\ref{eq:udpTranUtil}, $\mu_0(\Delta T)$ is decreasing or increasing with $\Delta T$, depending on whether $x_{02}$ is greater than $x_0$ or not, so the extreme is obtained when $\Delta T= \tau(b)$. In other words, the extreme (lowest or largest) values occur $\tau(b)$ after rate change. That means, the last packet of the old rate $x_0$ is transmitted with the extreme utilization $\mu_0^{*}=\mu_0(\tau(b))$.  The next theorem gives the maximum or minimum value. First, we start with Lemma \ref{lemma:alpha1} which captures the special case when $\alpha=1$.

    \begin{lemma}\label{lemma:alpha1}
        For $\alpha=1$, $\mu_0(\tau(b))=\mu_0$.
    \end{lemma}

     \begin{theorem}\label{eq:Thorem_maxu(t)_0Alpha}
        Assume the current steady-state utilization $\mu_0$ when the input UDP rate is $x_0$. If $x_{02}=\alpha x_0$ at $t=0$, $\alpha \in[0,\infty)$, then the extreme (minimum / maximum) UDP utilization during the transient regime is given by
        \begin{equation}\label{eq:alphaU}
            \mu_0(\tau(b))= \bigg[1+a\bigg(\frac{1-\mu_0}{a\mu_0}\bigg)^{\alpha}\bigg]^{-1}.
        \end{equation}
    \end{theorem}

    See the Appendix for proofs of Theorem \ref{eq:Thorem_maxu(t)_0Alpha} and Lemma \ref{lemma:alpha1}.
    \newtheorem{corollary}{Corollary}
     \begin{corollary}\label{cor:alpha0}
        For $\alpha=0$, or when the UDP flow stops, $\mu_0(\tau(b))=\rho_0(0)$.
    \end{corollary}

    Theorem \ref{eq:Thorem_maxu(t)_0Alpha} is a key finding of this paper and can be visualized using Fig.~\ref{fig:uChokeTransientTheory}. Note that the figure illustrates both steady-state and (extreme) transient UDP utilizations for the selected UDP arrival rates. When $\alpha=1$, by Lemma \ref{lemma:alpha1}, the values shown are the steady-state utilization $\mu_0$ for the given UDP arrival rate $x_0$.  When the UDP arrival rate $x_0$  changes (i.e., $\alpha \neq 1$),  the graph shows how far the UDP utilization can go up/down when  $x_0$ abruptly decreases/increases, respectively, to $x_{02}$. For example, assume an initial UDP rate of $x_0=2C$. The utilization for this input is read from the figure at $\alpha=1$ (also from Fig.~\ref{fig:olm-steadystate}) as $\mu_0=25.0\%$. When $x_{02}=0.2C$ which corresponds to rate change by a factor of $\alpha=0.1$, the transient utilization surges to $\mu_0(\tau(b))=56.5\%$. If the flow stops, by Corollary \ref{cor:alpha0}, the utilization instead jumps to $\rho_0(0)=60\%$ (see also Fig.~\ref{fig:rhoyVSy}). Similarly, when a flow of initial UDP arrival rate $x_0=3C$ stops, the transient utilizations can surge to a whopping $67\%$  from the initial $21\%$. 

    \begin{figure}[thb!]
            \centering
            \includegraphics[width=0.485\textwidth]{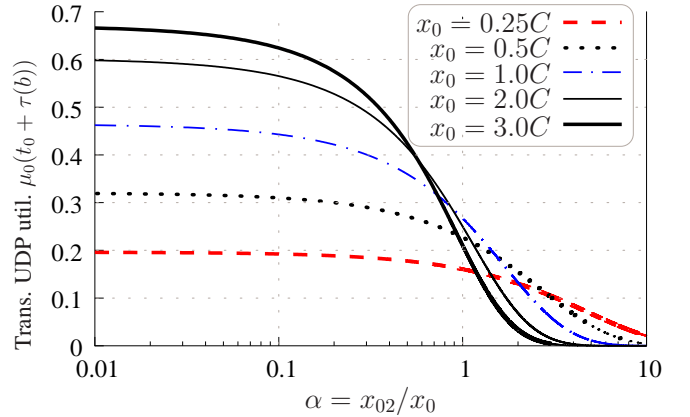}
            \caption{The impact on extreme UDP utilization of rate factor $\alpha$. Five previous inputs $x_0 \in\{0.25C,0.5C,1C, 2C, 3C\}$ are shown.}
            \label{fig:uChokeTransientTheory}
    \end{figure}

    \section{Evaluation}\label{sec:evaluation}

    In this section, we validate the results using simulations performed in ns-2.34. The network setup shown in Fig.~\ref{fig:sysModel} with the following settings is used: $C=20Mbps$ or 2500 pkt/sec, link latency $1ms$, buffer size $1000$ packets, $N=100$ TCP flows each of type SACK, RED buffer thresholds (in packets) $\min_{th}=20$ and $\max_{th}=1000$. Packet sizes are 1000 bytes. Flows start randomly on the interval [0,2] sec.

    We conducted extensive experiments,  each simulation replicated 500 times if not otherwise highlighted. The 95\% confidence intervals are so small that they are not reported. We remark that in computing the simulation results, unless otherwise stated, we have used a time window of 1ms. Since $C$=2500 pkts/sec, this is 2.5 times more than the per packet transmission interval assumed by the model, but the error due to this disparity is small and can be ignored.
    Section \ref{subsec:modelValidate} presents the validation of the model, and Section \ref{subsec:furtherResults} presents additional simulation results.

    \subsection{Model Validation}\label{subsec:modelValidate}
    In this section, we validate the two important results of this paper: Theorem \ref{eq:Thorem_maxu(t)_0Alpha} and Lemma ~\ref{eq:udpTranUtil}.

    \subsubsection{Validation of Theorem~\ref{eq:Thorem_maxu(t)_0Alpha}}
    In Fig.~\ref{fig:uChokeTransient}, we show for selected initial UDP arrival rates the impact of rate change by factor $\alpha\in[0.01,10]$. As can be seen, the simulation results accurately match the model predictions. For instance, for $x_0=3C$ and $x_{02}=0.03C$, or $\alpha=0.01$, the maximum utilizations obtained by the model and simulations are 67\% and 65\%, respectively.

    \begin{figure}[thb!]
            \centering
            \includegraphics[width=0.485\textwidth]{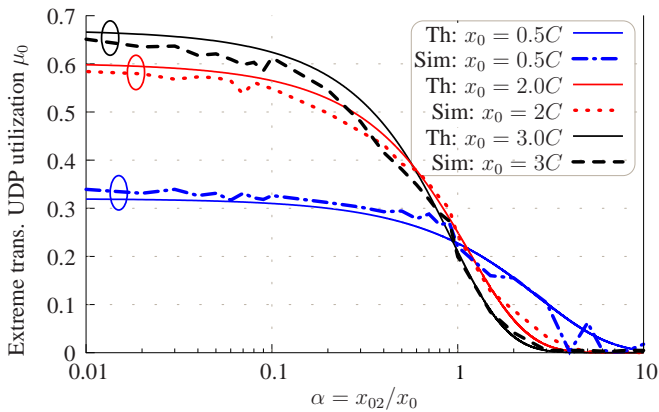}
            \caption{Validation of extreme utilization stated by Th.~\ref{eq:Thorem_maxu(t)_0Alpha}.}
            \label{fig:uChokeTransient}
    \end{figure}

\vspace{5mm}

    \subsubsection{Validation of Lemma~\ref{eq:udpTranUtil}}
    For rate factors of $\alpha=5,~1/5,~10, ~1/10$, Fig.~\ref{fig:validateTh1} shows the evolution of transient UDP utilizations obtained through simulation and the analytical model stipulated by Lemma~\ref{eq:udpTranUtil}. The steady state backlog size $b$ required for the theoretical plot is taken from the steady state simulation just before rate change.

    \begin{figure}[htb!]\centering
          \subfloat[$x_0=0.25C$, $\alpha=5$]{\label{fig:validatealpha5X}\includegraphics[width=0.4\textwidth]{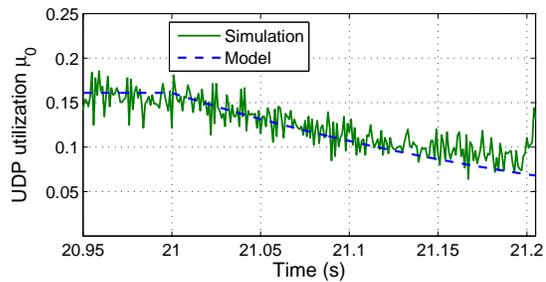}}
          \vspace{2mm}
          \subfloat[$x_0=3C$, $\alpha=1/5$]{\label{fig:validatealpha02X}\includegraphics[width=0.4\textwidth]{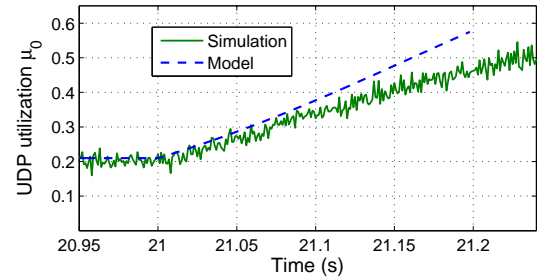}}
          \vspace{2mm}
          \subfloat[$x_0=0.25C$, $\alpha=10$]{\label{fig:validatealpha10X}\includegraphics[width=0.4\textwidth]{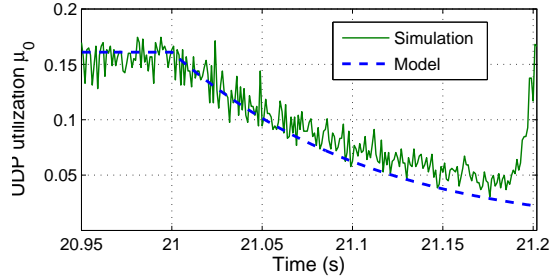}}
          \vspace{2mm}
          \subfloat[$x_0=3C$, $\alpha=1/10$]{\label{fig:validatealpha01X}\includegraphics[width=0.4\textwidth]{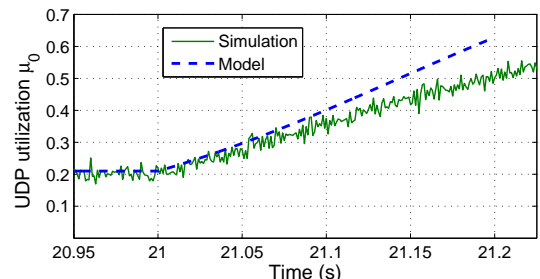}}
          \caption{Validating the transient utilization Eq.~\eqref{eq:newThinnedu_02}}
          \label{fig:validateTh1}
     \end{figure}

    As discussed earlier, there may be two sources of approximation errors for the theoretical results. First, the conservative assumption of constant backlog size during the transient phase. Second, the difference in measurement intervals. The model results are tallied per packet transmission time ($0.4$ms), while simulation results are based on $1$ms interval. Despite these differences, the model and simulation results are reasonably matching, even more so for moderate initial $x_0$ (see Fig.~\ref{fig:validatealpha5X}). In all figures, the approximation errors are negligible at the beginning. As we move further in time during the transient phase, however, the assumption of constant buffer occupancy $b$ fails to hold. Consequently, the errors become larger.  In the figure, the model can be treated as a bound of the respective transient utilization. Additionally, the extreme utilizations during the transient regime (i.e., the lowest utilization in Figs.~\ref{fig:validatealpha5X} and \ref{fig:validatealpha10X}, and highest utilization in Figs.~\ref{fig:validatealpha02X} and \ref{fig:validatealpha01X}), which can also be verified from the generic utilization curve shown in Fig.~\ref{fig:uChokeTransient}, are very close to the simulation results. 

    \subsection{Miscellaneous Results}\label{subsec:furtherResults}

     The next three subsections present additional results. The first two return to the motivational examples discussed in Sec.~\ref{sec:motivationBackgnd}. The results in Sec.~\ref{subsub:web} are based on a different traffic model---Web Traffic.

    \subsubsection{Results on Example 1}
    For the two experiments of Example 1 discussed in Sec.~\ref{subsec:motivation}, Table~\ref{tab:exp1} tabulates three sets of extreme UDP utilization values:  theoretical values based on Theorem~\ref{eq:Thorem_maxu(t)_0Alpha};  and two sets of simulation values averaged over measurement windows of $0.4$ms and $10$ms. Note that since transient utilization is continuously changing, the time granularity of measurement windows are critical (see also Sec.~\ref{subsub:exObservation}). In Fig.~\ref{fig:4x12xUtil}, the curves are based on a $10$ms window. Compared with simulation results,  the theoretical values in Table~\ref{tab:exp1} seem to represent \emph{lower} utilization bound when  $\alpha>1$ and \emph{upper} utilization bounds when $\alpha<1$. Nevertheless, the theoretical and simulation values based on packet transmission time window are remarkably close, and the bounds are tight. 

    \begin{table}[th]\centering
        \caption{Lowest and highest transient UDP utilizations for the experiments of Example 1 in Sec.~\ref{sec:motivationBackgnd}.}
        \label{tab:exp1}
        \begin{tabular}{  l  l  r  | r  r  r }
        \hline
        \multicolumn{3}{ c }{Scenarios}                                                                           &  \multicolumn{3}{| c }{Extreme UDP utilization}     \\
        \hline
        \multirow{2}{*}{$x_0/C$} &   \multirow{2}{*}{$x_{02}/{C}$} &      \multirow{2}{*}{$\alpha$}      & \multirow{2}{*}{ Model}     &\multicolumn{2}{ c } {Simulation}  \\
        \cline{5-6}
                                 &                               &                                     &                    & W=0.4ms       & W=10ms                \\
        \hline
        0.50                     &                     2.0       &          4                           &    6.7\%            &     7.5\%       &  10.5\%               \\
        2.0                      &                     0.50      &          1/4                         &    50.8\%           &     49.3\%      &  44.5\%                \\
        0.25                     &                     3.0       &          12                          &    1.3\%            &      2.0\%      &  3.8\%              \\
        3.0                      &                     0.25      &          1/12                        &    63.2\%           &      61.4\%     &  56.7\%                \\
        \hline
        \end{tabular}
    \end{table}

    \subsubsection{Results on Example 2}\label{subsubsec:example2}
     Next, for the motivating example in Sec.~\ref{subsec:motivation} where the UDP arrival rate alternates between $1$C and $10$C, the two extreme utilization values using Theorem~\ref{eq:Thorem_maxu(t)_0Alpha}  are found to be 0.015\% (for change from $1$C to $10$C) and 75\% (for change from $10$C to $1$C). Similar to what has been observed in Example 1, the extreme values are tight bounds. Even using a gross measurement window of 10ms for Fig.~\ref{fig:x0flap1C10C}, we still observe a peak utilization of 72\%. Besides the extreme values, the model allows us to explore the average utilization in a time period within the transient regime. To demonstrate this, Table~\ref{tab:exp2} shows average utilizations (per every $500$ms) using four methods: (1) steady state utilization corresponding to average UDP arrival rate $5.5$C, (2) average of steady state utilizations corresponding to UDP arrival rates $1$C and $10$C, (3) average utilization based on the model, i.e.,  Lemma~\ref{eq:udpTranUtil}, and (4) simulation results, using measurement time windows of $1$ms. Steady state values for the first two methods are based on the OLM model described in Sec.~\ref{subsub:olmModel}. For using Lemma~\ref{eq:udpTranUtil}, the values for backlog size $b$ at each 250ms interval are required. Our extensive simulations show that the backlog size is wildly changing over time. We took the backlog size at the onset of rate change and fed it as the input $b_0$ into  Lemma~\ref{eq:udpTranUtil}. For example, at $t=21$ $b_{0}=765$ and at $t=21.25$  $b_0=675$.  As the table shows, the steady-state analytical bounds are far from the simulation results. On the contrary, the transient analysis  nicely represents the picture of the queue even under radically changing traffic conditions.

    \begin{table}[th]\centering
        \caption{Average UDP utilizations for Example 2 in Sec.~\ref{sec:motivationBackgnd}.}
        \label{tab:exp2}
        \begin{tabular}{  l | r  r  r  r }
        \hline
        Time            &      \multicolumn{4}{ c }{Average UDP utilization in \%} \\
        Interval                               & $\mu_{0,5.5C}$ &  $\frac{\mu_{0,1C}+\mu_{0,10C}}{2}$    & Model &  Simulation   \\
        \hline
        $[21,21.5)$                       & 11.7   & 14.8  & 17.6  &   18.9              \\
        $[21.5,22)$                      &  11.7   & 14.8  & 19.1  &   19.4              \\
        $[22,22.5)$                      &  11.7   & 14.8  & 17.6   &  18.9              \\
        $[22.5,23)$                      &  11.7   & 14.8  & 18.3  &   19.4              \\
        $[23,23.5)$                      &  11.7   & 14.8  & 18.0  &   19.3              \\
        \hline
    \end{tabular}
    \end{table}

    \subsubsection{Results using Web traffic}\label{subsub:web}
    Since the Internet flow dynamics is heavily shaped by short Web transfers, we conducted a 500-replicated experiment explained below. The UDP arrival pattern is the same as in Fig.~\ref{fig:x0flap1C10C} except that $x_0=10C$ for $t< 21$ and $x_0=1C$ for  $ t> 23$. The Web traffic is modeled as follows: Starting from $t=20$s, each of the 100 TCP sources (see Fig.~\ref{fig:sysModel}) generates a Poisson process with an average arrival rate of 25 Hz. The size of each session (file) is Pareto-distributed with average size of 10kB (about 10 packets) and a shape parameter of 1:3. This model captures the heavy tailed nature of Web file sizes and their transmission times~\cite{selfSimilarity97}. Simulation lasts for 25 seconds and over 9000 Web sessions have been generated. The result is illustrated in Fig.~\ref{fig:webUtilization}.

    \begin{figure}[thb!]
            \centering
            \includegraphics[width=0.485\textwidth]{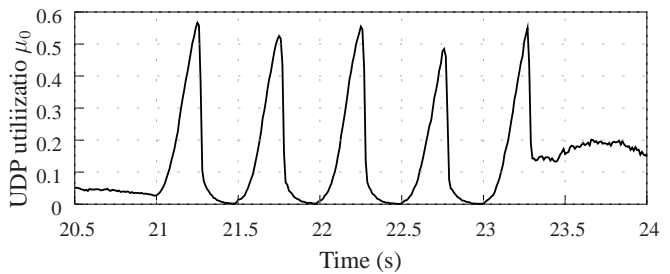}
            \caption{UDP utilization in the presence of Web flows.}
            \label{fig:webUtilization}
    \end{figure}


    Due to the huge and highly bursty  Web traffic generated, the buffer  is always full. Like in the long-lived TCP scenarios, UDP exhibits widely fluctuating throughput patterns during transient regimes. However, the extreme points of transient regime are generally lower in  value. For example, UDP utilization can get as low as 0.08\%, and as high as 56\%. Nevertheless, in both cases, they are bounded, albeit \emph{loosely} by the analytical extreme values 0.015\% and 72\% as discussed in Sec.~\ref{subsubsec:example2}. We believe that the smaller values are due to higher ambient drop rates $r$ (cf. $r$ in \eqref{eq:neglectRED}) caused by RED dealing with persistent full buffer occupancy. To close the difference, further study may be conducted where the rate conservation argument could be exploited.


    The results in this section show that the analytical results and observations made in this paper may apply to a wider context than studied here.

     \section{Conclusion}\label{sec:conclusion}

     While existing works on CHOKe reveal interesting structural, asymptotic and limit behaviors of the queue, their results are limited to the steady state when the queue reaches equilibrium in the presence of many long-lived TCP flows and constant rate UDP flows. Unfortunately, they lack showing properties of the queue in a possibly more realistic network setting where the exogenous rates of unresponsive flows may be dynamically changing, and consequently the model parameters, rather than being static, may be continuously evolving.

     This paper provides the first study on CHOKe behavior in the aftermath of rate changes in UDP traffic arrival. In particular, we are concerned with CHOKe queue behaviors during the transient regime which we model as a transition from one steady queue state to another. We found that the performance limits stipulated in the steady state rarely hold for such transient regimes. Depending on the nature of rate change, the queue exhibits instant fluctuations of UDP bandwidth sharing in reverse direction. This behavior has ramifications on the smooth operation of the Internet where most flows are rate-adaptive. Such flows may see fluctuating available link bandwidth and degrade in performance. By extending and leveraging the spatial distribution model, this paper analytically (1)  determines the extreme points of UDP utilization (observed within an order of queueing delay after rate change), and (2) tracks the  evolution of the transient UDP utilization following rate change. In addition, the model allows us to obtain generic UDP utilization plots that help explain both the transient extreme characteristics and steady-state characteristics. The analytic results have been rigorously validated through extensive simulations. We believe, the analytical approach used in this paper also sheds light on studying transient behaviors of other leaky queues.

    \bibliographystyle{IEEEtran}

    \bibliography{ChokeTransient-final}

\newpage
    \appendix

     This appendix is devoted to proving the lemmas and theorems in Sec.~\ref{subsec:modifiedSDM}.

     We need the following intermediate results.

     From \eqref{eq:qDelay} and \eqref{eq:v(y)}, respectively, we get
       \begin{align}
            \tau^{'}(y) &= 1/v(y) \label{eq:diffTau} \\
            v^{'}(y)/v(y)&=\rho^{'}_0(y)/(1-\rho_0(y)). \label{eq:diffVyVy}
       \end{align}
     Using $\rho_0^{'}(y)$ from \eqref{eq:diffVyVy} into \eqref{eq:lnThinningv1} and solving for $v^{'}(y)$,
        \begin{equation}
            v^{'}(y)= x_0(1-r)\beta \rho_0(y). \label{eq:diffVy}
        \end{equation}

    \subsection*{Proof of Lemma \ref{lemma-rho0(y)}}
    \vspace{3mm}

    \begin{proof}
    The following set of equations are trivial and follow from \eqref{eq:lnThinningv2}.
    \begin{align}
         x_0(1-r) \beta \int\limits_{0}^y\frac{1}{v(s)}ds &=  \int\limits_{0}^y \frac{\rho_0^{'}(s)}{\rho_0(s)}ds+ \int\limits_{0}^y\frac{\rho_0^{'}(s)}{1-\rho_0(s)}ds \nonumber
    \end{align}
    \begin{align}
            x_0(1-r) \beta \tau(y) &= \ln(\rho_0(s))|_{0}^{y}-\ln(1-\rho_0(s))|_{0}^{y} \nonumber \\
            x_0(1-r)\beta \tau(y) &=\ln\bigg[\frac{\rho_0(y)}{\rho_0(0)}\frac{1-\rho_0(0)}{1-\rho_0(y)}\bigg] \nonumber \\
            \tau(y)&=\frac{1}{x_0(1-r)\beta} \ln\bigg[\frac{\rho_0(y)}{1-\rho_0(y)}\frac{1-\rho_0(0)}{\rho_0(0)}\bigg] \nonumber \\
            \tau(y)&= \frac{1}{x_0(1-r)\beta} \ln\bigg[a\frac{\rho_0(y)}{1-\rho_0(y)}\bigg] \label{eq:lnThinningv3}
    \end{align}

From \eqref{eq:lnThinningv3}, $\rho_0(y)$ can be expressed as \eqref{eq:rho_0y}.
\end{proof}




    \subsection*{Proof of Lemma~\ref{lemma:intermediateR}}
    \vspace{3mm}
     \begin{proof}
     From \eqref{eq:invariance4} or \eqref{eq:olm_input}, practical values of $h_0$ must satisfy $h_0 < 0.5$.  For $h_0=0$ and $\mu_0=0$, the proof is trivial. We only need to prove the bounds of $\frac{1-\mu_0}{1-h_0}$ for $h_0 \in(0, 0.5)$.

     We need \eqref{eq:olm_u0} and the well known property of natural logarithms shown next.
     \begin{align}
            \ln x &\leq x-1             \qquad      \mbox{ for $x >0$}    \label{eq:naturalLogBound}
     \end{align}

     The proof is by contradiction and has two parts.
     \begin{enumerate}[(i)]
     \item \label{p1} First, we establish that $(1-\mu_0)/(1-h_0)\geq 1$. Let us assume:
     \begin{equation}\label{eq:h0LTu0}
         \frac{1-\mu_0}{1-h_0} < 1 \qquad  \Rightarrow \qquad h_0 < \mu_0.
     \end{equation}

    From \eqref{eq:h0LTu0} and \eqref{eq:olm_u0}, we get
    \begin{equation}\label{eq:hoLTu0v2}
        h_0 < \frac{\ln[\frac{1-h_0}{1-2h_0} ]}{ [ \frac{1-h_0}{1-2h_0} ]+ \ln[\frac{1-h_0}{1-2h_0} ] }
    \end{equation}

    \newpage

     Since $0<h_0<0.5$, $(1-h_0)/(1-2h_0)>1$ and $\ln[(1-h_0)/(1-2h_0)]> 0$. Multiplying both sides of \eqref{eq:hoLTu0v2} by  the denominator term found on the r.h.s. and simplifying, we obtain:
     \begin{equation}\label{eq:contradition}
       \frac{h_0}{1-2h_0} < \ln\frac{1-h_0}{1-2h_0}.
     \end{equation}
     On the other hand, since $\frac{1-h_0}{1-2h_0}>0$, we can apply property \eqref{eq:naturalLogBound} on $\frac{1-h_0}{1-2h_0}$ to obtain:
     \begin{equation}
       \ln\frac{1-h_0}{1-2h_0} \leq \frac{1-h_0}{1-2h_0} -1 = \frac{h_0}{1-2h_0}
     \end{equation}
     which is a contradiction to \eqref{eq:contradition} .
    Hence, we prove that $\frac{1-\mu_0}{1-h_0}\geq 1$. $\frac{1-\mu_0}{1-h_0}= 1$ when both $\mu_0,~h_0=0$.

     \item \label{p2} Here we establish that $(1-\mu_0)/(1-h_0)\leq 2$. As before, let us contradict by assuming that,
         \begin{equation}\label{eq:m0h0UpperBound}
         \frac{1-\mu_0}{1-h_0} > 2  \qquad  \Rightarrow \qquad \mu_0 < 2h_0-1.
     \end{equation}

     Since $h_0<0.5$, \eqref{eq:m0h0UpperBound} says that $\mu_0<0$, which is not possible. This means the assumption in \eqref{eq:m0h0UpperBound} must be wrong, or that $(1-\mu_0)/(1-h_0)\leq 2$.

     Combining (\ref{p1}) and (\ref{p2}) completes the proof.
     \end{enumerate}
   \end{proof}

    \subsection*{Proof of Lemma~\ref{eq:lemmaRho_0Y}}
    \vspace{3mm}
     \begin{proof}
     Solving for $\rho_0^{'}(y)$ from \eqref{eq:diffVyVy}  and \eqref{eq:diffVy},
        \begin{align}
            \rho_0^{'}(y)   &= x_0(1-r)\beta \bigg[\frac{\rho_0(y)-\rho_0(y)^2}{v(y)} \bigg]            \label{eq:diffRhoyv1}
        \end{align}
        Since $\beta < 0$, and probabilities $\rho_0(y)< 1$, $\rho_0(y)$ is decreasing with $y$. Taking the differentiation further and using \eqref{eq:diffVy},\eqref{eq:diffRhoyv1} in place of $v^{'}(y)$ and $\rho_0^{'}(y)$ and simplifying, we obtain
        \begin{align}
            \rho_0^{''}(y)  &= x_0(1-r)\beta \frac{\rho_0^{'}(y)v(y)(1-2\rho_0(y))-v^{'}(y)\rho_0(y)(1-\rho_0(y)) }{v^2(y)} \nonumber \\
                            &= x_0^2(1-r)^2\beta^2 \frac{\rho_0(y)(1-\rho_0(y))(1-3\rho_0(y))}{v^2(y)} \label{eq:diffdiffRhoy}
        \end{align}
    The critical point $(y^*,\rho_0(y^*))$ where $\rho_0^{''}(y)=0$ is given as,
    \begin{align}
        \rho_0^*    &=\rho_0(y^*)= \frac{1}{3} \label{eq:criticalRho_0Y} \\
        y^*         &= \frac{1}{K} \ln\bigg(\frac{a}{2}\bigg)+\frac{1}{K} \frac{1-3\rho_0(0)}{2(1-\rho_0(0))}. \label{eq:criticalY}
    \end{align}

    \eqref{eq:criticalY} is obtained upon substituting $\rho_0^*$ for $\rho_0(y)$ in \eqref{eq:qPositionYv1}.
    It is easy to see that $\rho_0(y)$ decreases in concave fashion on $y\in [0,y^*]$ and in convex fashion on $y\in[y^*,b]$. See Fig.~\ref{fig:rhoyVSy} for an example.

    Note that the critical point ($y^*, \rho_0^*$) exists when the UDP arrival rate $x_0$ exceeds a certain value $x_0^*$, calculated using \eqref{eq:rho_0(0)} as,
    \begin{align}\label{eq:minimXcriticalPoint}
        x_0^*&= \min\{x_0~ |~ \rho_0(0) \geqslant \rho_0^*\}= \min\bigg[ \frac{C}{2}\frac{1-\mu_0}{1-h_0}\bigg]= \frac{C}{2}.
    \end{align}

    Above, we use Lemma \ref{lemma:intermediateR} to state that $\frac{1-\mu_0}{1-h_0} \geqslant 1$. From \eqref{eq:minimXcriticalPoint}, it follows that when $x_0 \leqslant C/2$, $\rho_0(y)$ is strictly convex decreasing.

     For arrival rate $x_0\in[C/2,\infty)$, since $\rho_0(0)\geqslant \rho_o^*$ by \eqref{eq:minimXcriticalPoint} and $\rho_0(b):=\mu_0 \leqslant 26.9\% \leqslant \rho_0^*$ by the Limit property (see Sec.~\ref{sec:intro}), the critical point exists somewhere $y^*\in[0,b]$.

     \end{proof}

    \subsection*{Proof of Lemma~\ref{eq:lemmavY}}
    \vspace{3mm}
    \begin{proof}
    From \eqref{eq:diffVy}, it is trivial to see that $v(y)$ is  decreasing with $y$ since  $v^{'}(y)<0$. Differentiating \eqref{eq:diffVy} and using \eqref{eq:diffRhoyv1},
    \begin{align}
        v^{''}(y)&=x_0^2(1-r)^2\beta^2 \rho_0(y)\frac{1-\rho_0(y)}{v(y)}
    \end{align}
    Since $v''(y)>0$, $v(y)$ is convex decreasing throughout the queue. See Fig.~\ref{fig:vyVSy}.
    \end{proof}

    \subsection*{Proof of Lemma~\ref{eq:lemmaTauY}}
    \vspace{3mm}
     \begin{proof}
     The queueing delay is a strict convex increasing function, since from \eqref{eq:diffTau} and \eqref{eq:diffVy}, respectively,
      \begin{align}
          \tau^{'}(y) &= 1/v(y) > 0 \label{eq:diffTauv2} \\
          \tau^{''}(y) &= -\frac{1}{v^2(y)}v^{'}(y)=-\frac{x_0(1-r)\beta \rho_0(y)}{v^2(y)}>0 \label{eq:diffdiffTau}
      \end{align}
      \end{proof}

      \subsection*{Proof of Lemma~\ref{lemma:alpha1}}
      \vspace{3mm}
     \begin{proof}
    Set $\Delta T= \tau(b)$ and $\alpha=1$, i.e., $x_{02}=x_0$ in Lemma~\ref{eq:udpTranUtil}.
    \begin{subequations}
    \begin{align}
        \mu_0(\tau(b))|_{\alpha=1}&= \frac{1}{1+ae^{-x_0\tau(b)\beta}}=\frac{1}{1+ae^{-x_0\tau(b)\ln(1-1/b)}} \nonumber \\
        &=\frac{1}{1+a(1-1/b)^{-x_0\tau(b)}} \label{subeqn1} \\
        &=\frac{\rho_0(0)}{\rho_0(0)+(1-\rho_0(0))(1-1/b)^{-x_0\tau(b)}} \label{subeqn2} \\
        &=\frac{\rho_0(0)v(0)(1-1/b)^{x_0\tau(b)}}{\rho_0(0)v(0)(1-1/b)^{x_0\tau(b)}+N\rho_1(0)v(0)} \label{subeqn3} \\
       &=\frac{\rho_0(b)v(b)}{\rho_0(b)v(b)+N\rho_1(0)v(0)}=\mu_0  \label{subeqn4}
    \end{align}
    \end{subequations}
    In the above step, we use $\rho_0(b)=\mu_0$, $v(b)=C$ from \eqref{eq:icAtb}, and $N\rho_1(0)v(0)=(1-\mu_0)C$ from \eqref{eq:tcpThroughput}.
    \subsection*{Proof of Theorem~\ref{eq:Thorem_maxu(t)_0Alpha}}
    Similar to proof of Lemma \ref{lemma:alpha1}, we proceed for the general $\alpha$ as,
    \begin{equation}
        \mu_0(\tau(b))=\frac{1}{1+a(1-1/b)^{-\alpha x_0\tau(b)}} \label{eq:gmu_0}
    \end{equation}
    But from proof of Lemma~\ref{lemma:alpha1} above, we find for $\alpha=1$ that $\mu_0=1/(1+a(1-1/b)^{-x_0\tau(b)})$. After rearranging,  we obtain $(1-1/b)^{-x_0\tau(b)}=(1-\mu_0)/a\mu_0$. Substituting this into \eqref{eq:gmu_0} for general $\alpha$ completes the proof.
    \end{proof}

\end{document}